\documentclass[journal,draftcls,onecolumn,12pt,twoside]{IEEEtran}
\usepackage{cite}
\usepackage{amsmath,amssymb,amsfonts,amsthm}
\usepackage{algorithmic}
\usepackage{graphicx}
\usepackage{multirow}
\usepackage{bm}
\usepackage{enumitem}
\usepackage[normalem]{ulem} 
\def\dout{\bgroup
 \markoverwith{\lower-0.2ex\hbox
 {\kern-.03em\vbox{\hrule width.2em\kern0.45ex\hrule}\kern-.03em}}%
 \ULon}
\MakeRobust\dout
\usepackage{setspace}
\usepackage{diagbox}
\usepackage{float}
\usepackage{lipsum}
\usepackage{algorithm} 
\usepackage{textcomp}
\usepackage[font=small,skip=0pt]{caption}
\usepackage{subfigure}
\usepackage{xcolor}
\allowdisplaybreaks
\usepackage{balance}    
    
\newtheorem{theorem}{Theorem}
\newtheorem{definition}{Definition}
\newtheorem{lemma}{Lemma}

\begin{document}
\title{On the Fundamental Limits of Coded Caching Systems with Restricted Demand Types
}

\author{Shuo Shao, Jesús Gómez-Vilardebó, Kai Zhang, and Chao Tian\thanks{S. Shao is with Shanghai Jiao Tong University, SEIEE. (e-mail:shuoshao@sjtu.edu.cn)

J. Gómez-Vilardebó is with Centre Tecnològic de Telecomunicacions de Catalunya (CTTC/CERCA). 

K. Zhang was with Texas A\&M University, Department of Electrical and Computer Engineering during this work and he is now with UTHealth School of Biomedical Informatics. 

C. Tian are with Texas A\&M University, Department of Electrical and Computer Engineering. 

This paper is published on IEEE Transaction on Communications, doi: 10.1109/TCOMM.2020.3038392.

The work of Shuo Shao was supported
by NSFC-61872149, NSFC-61901261 and 19YF1424200. The work of
J. Gómez-Vilardebó was supported in part by the Catalan Government under
Grant GR2017-1479, and in part by the Spanish Government under Grant
TI2018-099722-B-I00 (ARISTIDES). The work of K. Zhang and C. Tian
was supported in part by the National Science Foundation under Grants CCF-18-32309 and CCF-18-16546. }}


\maketitle
\begin{abstract}
 Caching is a technique to reduce the communication load in peak hours by prefetching contents during off-peak hours. 
An information theoretic framework for coded caching was introduced by Maddah-Ali and Niesen in a recent work, where it was shown that significant improvement can be obtained compared to uncoded caching. Considerable efforts have been devoted to identify the precise information theoretic fundamental limits of the coded caching systems, however the difficulty of this task has also become clear.
One of the reasons for this difficulty is that the original coded caching setting allows all possible multiple demand types during delivery, which in fact introduces tension in the coding strategy. In this paper, we seek to develop a better understanding of the fundamental limits of coded caching by investigating systems with certain demand type restrictions. We first consider the canonical three-user three-file system, and show that, contrary to popular beliefs, the worst demand type is not the one in which all three files are requested. Motivated by these findings, we focus on coded caching systems where every file  must be requested by at least one user. A novel coding scheme is proposed, which can provide new operating points that are not covered by any previously known schemes. 
\end{abstract}

\section{Introduction}
Caching is a technique to alleviate communication load during peak hours by prefetching certain contents to the memory of the end users during off-peak hours. 
Maddah-Ali and Niesen \cite{MAN-14} proposed an information theoretic framework for caching systems, and showed that coded caching strategies can achieve significant improvement over uncoded caching strategies.
The system in the proposed framework, which has $N$ files and $K$ users, operates in two phases: during the {\em prefetching phase}, each user fills the cache memory of size $M$ with information on the files, and during the 
{\em delivery phase}, the users first reveal their requests, then the central server broadcasts a message of size $R$ to all the users, and finally each user makes use of the received message together with the content in the cache memory to reconstruct the requested file.

The optimal tradeoff between $M$ and $R$ is of  fundamental importance in this setting, the characterization of which has attracted significant research effort. Several inner bounds have been obtained using strategies based on either uncoded prefetching or coded prefetching \cite{MAN-14,Yu-18,Shraei-16,T-C-18,Jesus-18,Chen-16,Kumar-19,Kumar-19-2,Amiri-17,Z-T-18,Cao-20}, and various outer bounds have also been discovered \cite{Tian-16,Tian-18,Yu-19,Ghasemi-17,Wang-18}. When the prefetched contents are required to be uncoded fragments of the original files, the optimal tradeoff between $M$ and $R$ was fully characterized in \cite{Yu-18}. However, it was also shown that in general, optimal tradeoff requires coded prefetching strategies \cite{MAN-14,Chen-16,Tian-18}. The fundamental limit of coded caching systems still remains largely unknown in general.

In the prefetching phase, the users have no prior knowledge on the demands in the delivery phase, the collection of which is jointly referred to as the demand vector. As such, the prefetched contents need to be properly designed to accommodate all possible demand vectors. In a recent work \cite{Tian-16} (see also \cite{Tian-18}), the notion of demand type was introduced to classify the possible demand vectors, which led to simplifications in studying outer bounds of the coded caching systems. From this perspective, the original setting of \cite{MAN-14} in fact allows all possible demand types, and it appears that one reason for the afore-mentioned difficulty is the tension among the coding requirements to accommodate these different demand types. Therefore, a natural question is how different demand types impact the optimal $(M,R)$ tradeoff. 

To develop a better understanding of this issue, in this work we consider caching systems with restricted demand types, where during the prefetching phase, the users and the server know a priori that the demand vector in the delivery phase must be from a certain class of demand types. Such systems clearly have a more relaxed coding requirement than the original setting due to the prior knowledge on the possible demands, however, it is still highly nontrivial since each demand type would allow a rich set of possible demand vectors. Because of the relaxed coding requirement, any scheme designed to accommodate all possible demand types is valid for a system with restricted demand types, however, outer bounds for a system allowing all possible demand types may not hold for a system with restricted demand types. 


We begin our study by first collecting the best known inner bounds and outer bounds in the literature for the canonical $(N,K)=(3,3)$ system, some of which are from very recent developments \cite{Jesus-18,Kumar-19,Kumar-19-2,Cao-20} in the area. This exercise reveals that although the setting where all files are requested by at least one user may pose a significant challenge for the optimal code design, a system where only such demand types are allowed can in fact achieve $(M,R)$ pairs that are strictly impossible for systems where certain other demand types are allowed. This is contrary to popular belief that such a demand type is the ``worst case", and also confirms the existence of a tension among the coding requirements for different demand types. 

Given the observation above, we focus on one class of systems with restricted demand types, where it is known a prior that every file is requested by at least one user (implying that $K\geq N$).  We propose a novel code construction which provides improved performance to the known memory-rate tradeoff in the literature for such systems. The proposed scheme can be viewed as a novel approach to utilize the interference alignment technique \cite{Cadambe:08}, where some prefetched and transmitted symbols are aligned together to cancel certain interfering signal, in order for the desired signals to be recovered. Moreover, a novel pairwise transformation is introduced in the construction, in order to more efficiently take advantage of such alignment opportunities. 
We note that the proposed construction generalizes the code recently proposed in \cite{Kumar-19}, however, in contrast to the code specifically designed for $N=K$ that yields a single $(M,R)$ pair in each case, our construction is for general $(N,K)$ where $K\geq N$, and it provides multiple new operating points for each case. 

\section{Caching Systems with Restricted Demand Types}


In an $(N,K)$ coded caching system, there are $N$ mutually independent uniformly distributed files $(W_1,W_2,\dots,W_N)$, each of $F$ bits. There are $K$ users, each with a cache memory of capacity $MF$ bits; the set of all users  $\{1,2,\dots,K\}$ is denoted as $\mathcal{K}$. In the prefetching phase, user-$k$ stores some content, denoted as $Z_k$, in the local cache memory. Rigorously, we can define the prefetching process as follow. 
\begin{definition}[Prefetching scheme] In an $(N,K)$ coded caching system, we define a prefetching scheme under a given cache memory capacity $M$ as $K$ caching encoding functions denoted as $\bm{\Phi}_M=(\phi_{M,1},\phi_{M,2},\dots,\phi_{M,K})$, where each encoding function $\phi_{M,k}:\{0,1\}^{NF} \rightarrow \{0,1\}^{\lfloor MF\rfloor}$ for $k=1,2,\dots,K$ maps the $N$ files to each user's cache memory content. That is to say, 
\begin{equation*}
    Z_k=\phi_{M,k}(W_1,W_2,\dots,W_N)
\end{equation*}
for $k=1,2,\dots,K$;
\end{definition}

In the delivery phase, user-$k$, $k\in \mathcal{K}$, requests file $d(k)$. In this work, we consider the case that the demand vector $\bm{d}= ( d(1), d(2),\dots, d(K))$ belongs to a given set $\mathcal{D}$. With the given prefetching scheme $\bm{\Phi}_M$ and demand vector $\bm{d}$, the central server broadcasts a message $X_{\bm{\Phi}_M,\bm{d}}$ of $RF$ bits to every user, such that, together with the cached content, each user can decode the requested file, where $R$ is the rate of communication. We give the following rigorous definition of the communication rate. 

\begin{definition}[Achievable broadcast rate]
In an $(N,K)$ coded caching system, a communication rate $R(\bm{\Phi}_M,\bm{d})$ is $\epsilon$-achievable for given prefetching scheme $\bm{\Phi}_M$ and demand vector $\bm{d} = ( d(1), d(2),\dots, d(K))$, if and only if there exists:
\begin{itemize}
\item An message encoding function $f_{\bm{\Phi}_M,\bm{d}}: \{0,1\}^{NF} \rightarrow \{0,1\}^{\lfloor R(\bm{\Phi}_M,\bm{d})F\rfloor}$ for given prefetching scheme $\bm{\Phi}_M$ and each possible demand vector $\bm{d}$ that maps the $N$ files into the broadcast messages based on $\bm{d}$, thus
\begin{equation*}
   X_{\bm{\Phi}_M,\bm{d}}=f_{\bm{\Phi}_M,\bm{d}}(W_1,W_2,\dots,W_N);
\end{equation*}

\item $K$ decoding functions $\psi_{\bm{\Phi}_M,\bm{d},k}: \{0,1\}^{\lfloor MF\rfloor}\times \{0,1\}^{\lfloor R(\bm{\Phi}_M,\bm{d})F\rfloor} \rightarrow \{0,1\}^F$ that map the caching content of user $k$ given by $\phi_{M,k}$ and the broadcast message to an $\epsilon$-error estimate of the requested file of user $k$ for all $k=1,2,\dots,K$, thus
\begin{align*}
    & Pr(W_{d(k)}\neq \psi_{\bm{\Phi}_M,\bm{d},k}(\phi_{M,k}(W_1,W_2,\dots,W_N),X_{\bm{\Phi}_M,\bm{d}}))\\
    & \leq \epsilon.\end{align*}
for all $k=1,2,\dots,K$.
\end{itemize}
\end{definition}

The communication rate $R$ of the $(N,K)$ coded caching systems, for a given cache memory capacity $M$ and demand vector set $D$, is defined as the rate of the best scheme who can minimize its worst case rate over all possible demand vectors. That is to say, the formal definition of $R$ is as follow.
\begin{definition}[Communication rate]
In an $(N,K)$ coded caching system with demand restriction set $\mathcal{D}$, the communication rate $R$ of the broadcast message in the delivery phase for the prefetching memory size $M$ is called achievable if 
\begin{align*}
    R\geq \sup_{\epsilon>0}\inf_{F} & \min_{\bm{\Phi}_{M'}:M'\leq M} \max_{\bm{d}\in \mathcal{D}} \inf\{ R(\bm{\Phi}_{M'},\bm{d}):\\
    & R(\bm{\Phi}_{M'},\bm{d}) \text{is $\epsilon$-achievable}\}.
\end{align*}
\end{definition}

In the coded caching problem, the optimal tradeoff between $M$ and $R$ is the fundamental mathematical object of interest. In order to obtain some insight of this optimal tradeoff, we would like to introduce the notion of "demand type".
It was first introduced in \cite{Tian-16} (see also \cite{Tian-18}), which is restated below.
\begin{definition}[Demand types]
For a demand vector $\bm{d} = ( d(1), d(2),\dots, d(K))$ in an $(N,K)$ coded caching system, denote the number of users requesting file $W_n$ as $e_n$, where $n \in [1:N]$.  The demand type of $\bm{d}=( d(1), d(2),\dots, d(K))$ is the length-$n$ vector obtained by sorting the values $\bm{e}(\bm{d})=(e_1,e_2,\dots ,e_N)$ in a decreasing order, which is denoted $\vec{\bm{e}}(\bm{d})$.
\end{definition}

In this work, we consider caching systems with restricted demand types, i.e., it is known a priori that the demand vector must belong to a class of demand types. As a special case, if the class contains only one demand type, it will be referred to as a {\em single demand type system}. In this case, the demand vector set $\mathcal{D}=\mathcal{D}_{s}\triangleq\{\bm{d}:\vec{\bm{e}}(\bm{d})= \vec{\bm{e}}\}$ for some given demand type $\vec{\bm{e}}$. A second special case is when the class is restricted to be the demand types where all files are requested by at least one user, which is referred to as {\em the fully demanded system}. In this case, the demand vector set $\mathcal{D}=\mathcal{D}_{fd}\triangleq\{\bm{d}:\vec{\bm{e}}(\bm{d})~\text{such that}~e_i(\bm{d})>0, i=0,1,\ldots,N\}$. A third special case, where all possible demand vectors are allowed, is the original setting considered in \cite{MAN-14}, which will be referred to as the {\em fully mixed demand type system}. In this case, the demand vector set $\mathcal{D}=\mathcal{D}_{mixed}\triangleq\{1,2\dots,N\}^{K}$.

As an example, consider an $(N,K)=(3,3)$ system. When the demand vector $( d(1), d(2), d(3))=(2,2,1)$,  we have $\bm{e}(\bm{d})=(1,2,0)$ and $\vec{\bm{e}}(\bm{d})=(2,1,0)$. Therefore this demand vector belongs to the demand type $\bm{e}(\bm{d})=(2,1,0)$. More generally, demand type $\vec{\bm{e}}=(2,1,0)$ includes the following possible demand vectors in Table \ref{demand vectors}, which induce 6 different $\bm{e}(\bm{d})$ triples.  

\begin{table}[h]
\centering
\caption{Demand vectors of demand type $(2,1,0)$ and the corresponding $\bm{e}(\bm{d})$}\label{demand vectors}
\scalebox{0.9}{
\begin{tabular}{|c|c|}
\hline
$\bm{e}(\bm{d})=(2,1,0)$ & $(1,1,2),(1,2,1),(2,1,1)$\\\hline
$\bm{e}(\bm{d})=(1,2,0)$ & $(2,2,1),(2,1,2), (1,2,2)$\\\hline
$\bm{e}(\bm{d})=(2,0,1)$ &$(1,1,3),(1,3,1),(3,1,1)$\\\hline
$\bm{e}(\bm{d})=(1,0,2)$ &$(3,3,1),(3,1,3), (1,3,3)$\\\hline
$\bm{e}(\bm{d})=(0,2,1)$ &$(2,2,3),(2,3,2),(3,2,2)$\\\hline
$\bm{e}(\bm{d})=(0,1,2)$ &$(3,3,2),(3,2,3)(2,3,3)$\\\hline
\end{tabular}}
\end{table}

\section{Mixed Demand Type and Single Demand Type Systems: The $(3,3)$ Case}

In this section, we consider the canonical $(N,K)=(3,3)$ system, and collect the best known inner bounds and outer bounds in the literature for different caching systems depending on the imposed demand type restrictions. This exercise reveals several important insights and fundamental differences between systems with different demand type restrictions. 

\begin{figure}[tb]
\setlength{\belowcaptionskip}{-0.3cm}
  \centering
  \includegraphics[width=0.8\linewidth,draft=false]{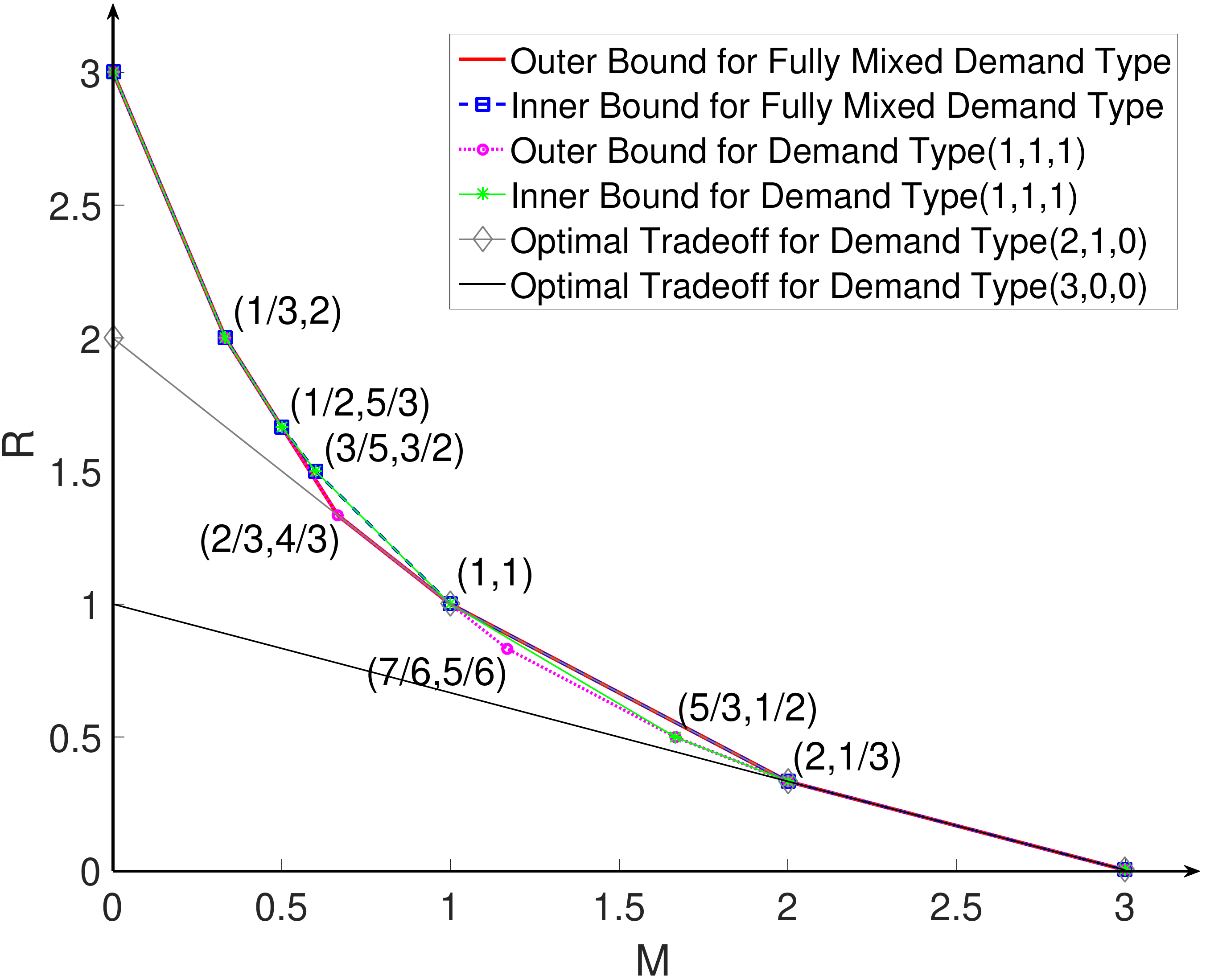}\\
  \caption{Inner bounds and outer bounds for various $(3,3)$ systems.\label{figouterboundmix}}
\end{figure}

\subsection{The Fully Mixed Demand Type System}

The best known outer bound for this system can be found in \cite{Tian-18}, which are all the non-negative pairs of $(M,R)$ satisfying the constraints
\begin{gather*}
3M+R\geq 3,  6M+3R\geq 8,  M+R\geq 2,\\
2M+3R\geq 5, M+3R\geq 3.
\end{gather*}
The best known inner bound, on the other hand, is given by the lower convex hull of the points
\begin{align*}
(0,3),(1/3,2),(1/2,5/3),(3/5,3/2),(1,1),(2,1/3),(3,0),
\end{align*}
where the second and third points are achieved by the scheme in \cite{Jesus-18}, the fourth point is achieved by that in \cite{Cao-20}, while the others can be achieved by that in \cite{MAN-14}.



\subsection{Single Demand Type Systems}

Next we provide the best known results for the three single demand type systems.
\begin{enumerate}
    \item For the system with the demand type $(3,0,0)$, i.e., the same file is requested by all three users, the achievable region is all the non-negative $(M,R)$ such that
    \begin{align}
        M+3R\geq 3,
    \end{align}
    i.e., in this case, the inner bound and the outer bound match. The outer bound can be obtained by a simple cut-set argument \cite{MAN-14}, while the inner bound is trivial through a memory-sharing argument.
    \item For the system with the demand type $(2,1,0)$, the achievable region is all the non-negative $(M,R)$ such that
    \begin{equation*}
    M+ R \geq  2,  2M+ 3R\geq 5, M+ 3R \geq 3.
    \end{equation*}
    In this case, the corner points $(1,1)$ and $(2,1/3)$ can be achieved using the scheme in \cite{MAN-14}, and the points $(0,2)$ and $(3,0)$ are trivial. The outer bound was established in \cite{Tian-18}.
    \item For the system with the demand type $(1,1,1)$, the best known outer bound is given in \cite{Tian-18} as
    \begin{gather*}
3M+ R \geq3,6M+ 3R  \geq 8, M+ R  \geq2, \\
12M+ 18R  \geq29, 3M+ 6R  \geq8, M+ 3R  \geq 3.
\end{gather*}
The best known inner bound is given by the lower convex hull of the points
\begin{gather*}
   (0,3), (1/3,2), (1/2,5/3), (3/5,3/2), (1,1), \\
   (5/3,1/2), (2,1/3), (3,0).
\end{gather*}
The second and third points are achieved by the scheme in \cite{Jesus-18}, the point $(3/5,3/2)$ by that in \cite{Cao-20}, the point $(5/3,1/2)$ by that in \cite{Kumar-19}, and the others by that in \cite{MAN-14}.
\end{enumerate}

\subsection{Mixed Demand Types vs. Single Demand Type}

By comparing the rate regions of different demand type systems in Fig.\ref{figouterboundmix}, we make the following observations:
\begin{enumerate}
\item The point $(5/3,1/2)$, which is achievable for the system with the single demand type $(1,1,1)$ as shown in \cite{Kumar-19}, is in fact not achievable for the $(2,1,0)$ demand type system, thus also not achievable for the fully mixed demand type system. 
\item Between mixed and single demand type systems, single demand type systems can indeed achieve lower rates than the fully mixed demand type system. 
\item Different single demand type systems provide different outer bounds for the fully mixed demand type system, with the one with fewer files demanded produce better bounds at high memory regimes, while those with more files being better at low memory regimes.
\end{enumerate}

The first observation implies that the case when all files are requested is not necessarily the ``worst case", contrary to popular belief in the coded caching literature. Thus designing codes for this demand type alone is not sufficient to yield the optimal scheme for the fully mixed demand type systems. In fact, the optimal code design for fully demanded systems, though by itself already a challenging problem, can even be ``simpler'' than other cases. Motivated by the observations above and the code construction proposed in \cite{Kumar-19}, which only provided a single operating point when $N=K$, in the sequel we focus on fully demanded systems (implying that $N\leq K$), for general $(N,K)$ parameters.

\section{Main Results}

 Given a demand type $\vec{\bm{e}}$ which does not contain any zero elements, i.e., each file is being requested by at least one user, let $p_{\vec{\bm{e}}}$ be the number of ones in it, i.e., the number of files requested by only one user. 
 The \textit{delivery rate compression saving factor}, $S(\vec{\bm{e}},r)$ for $r=\left\{ 0,...,K-1\right\}$, is defined as
\begin{align}
S(\vec{\bm{e}},r)=\frac{(K-p_{\vec{\bm{e}}})\binom{K-1-N}{r+1}+p_{\vec{\bm{e}}}\binom{K-1-(N-1)}{r+1}}{K\binom{K-1}{r}},    
\end{align}
where we have taken the convention that $\binom{k}{n}=0$ when $k\leq n$. The \textit{delivery rate compression saving factor} accounts for the rate reduction in the transmission rate when linear dependence may be found among transmitted symbols in the delivery phase. Observe that $S(\vec{\bm{e}},r)=0$ when $K=N$, and it generally increases when $K$ increases. 

The performance of the proposed scheme is summarized in the following theorem. 

\begin{theorem}
\label{thm1 iq channel} 
For an $(N,K)$ fully demand caching system (implying $N\leq K$) where only a single demand type $\vec{\bm{e}}$ is allowed, the following memory-rate pairs are achievable:
\begin{align}
\left(M,R(\vec{\bm{e})}\right)=\left(r\frac{N-1}{K-1}+\frac{r+1}{K},\frac{K-1-r}{r+1}-S(\vec{\bm{e}},r)\right)
\end{align}
for $r=\left\{ 0,...,K-1\right\}$.
\end{theorem}

The proof of theorem is given in the next section. It will be clear from the proposed code construction that the restriction on the demand types simplifies the design of the prefetching. More precisely, it is  effective to prefetch certain coded symbols involving all files, because any file will be requested eventually in such systems. 
 
 Since in the proposed code construction, the prefetching strategy does not depend on the particular demand type, but only requires that all files are requested during the delivery phase, we can state the following theorem. 

\begin{theorem}
\label{thm2 iq channel} 
For an $(N,K)$ fully demanded caching system (implying $N\leq K$), the following memory-rate pairs are achievable:
\begin{align}
\left(M,R\right)=\left(r\frac{N-1}{K-1}+\frac{r+1}{K},\frac{K-1-r}{r+1}-S(\vec{\bm{e}^{\ast}},r)\right)
\end{align}
where the demand type $\vec{\bm{e}^{\ast}}$ is any that satisfies $p_{\vec{\bm{e}^{\ast}}}=\max(2N-K,0)$. 

\end{theorem}
\begin{proof}
We will need to find the largest rate $R(\vec{\bm{e})}$ in Theorem \ref{thm1 iq channel}. The choice of $\vec{\bm{e}}$ that maximizes $R(\vec{\bm{e}})$ minimizes $S(\vec{\bm{e}},r)$, which in turn minimizes $p_{\vec{\bm{e}}}$, i.e., the number of files that are requested only by one user. If $K\geq2N$, we can set $p_{\vec{\bm{e}^{\ast}}}=0$ by choosing a demand type where every file is requested by at least two users. On the other hand, if $N\leq K<2N$, we minimize $p_{\bm{e}}$ by choosing a demand type where $N-p_{\vec{\bm{e}}}$ files are requested by one user and the other $N-p_{\vec{\bm{e}}}$ files are requested by two users, i.e., $p_{\vec{\bm{e}^{\ast}}}+2(N-p_{\vec{\bm{e}^{\ast}}})=K$. 
\end{proof} 

In an earlier version of this work \cite{SGZT-19}, we proposed a construction for $(N,K=N)$ systems, for which Theorem \ref{thm1 iq channel} provides a natural generalization. For $(N,K=N)$ fully demanded systems, by setting $r=K-2$, we recover the operating point given in \cite{Kumar-19}. In fact, our proposed construction in this work (sans the pairwise transformation) and that given in \cite{SGZT-19} essentially specialize to that in \cite{Kumar-19} by a proper relabeling of the file segments. On the other hand, setting $r=0$ gives the point in \cite{Chen-16} (see also \cite{T-C-18}). Operating points for other values of $r$ are previously unknown to be achievable.

The performance of the proposed construction is illustrated for $(N,K)=(4,6)$ and demand type $(3,1,1,1)$ in Fig. \ref{fig:46}; the operating points $(14/15,1.9)$, $(17/10,1)$, $(37/15,1/2)$ and $(97/30, 1/5)$ are previously unknown to be achievable in the existing literature. 

As mentioned earlier, there exists a tension among the coding requirements of different demand types. Since our coding scheme is designed specifically for fully demanded systems, its performance naturally suffers when some files are not requested. In practice, this implies that the proposed code is more suitable when  with high probability all files will be requested. When it does occur that some files are not requested,  a straightforward solution is simply to use a higher rate of delivery. More precisely, the server can additionally transmit some parity check symbols of the files that are not being requested, and it can be shown that the additional communication cost in upper-bounded by $\frac{K}{r+1}-S(\vec{\bm{e}^{\ast}},r)$. Since the probability of these events happening is extreme low, the overall system performance would not be impacted significantly.
In other settings where some files are not being requested more frequently, the proposed code is less suitable; nevertheless, at the expense of losing some performance, the virtual user technique given in \cite{Jesus-18-1} can be evoked to accommodate such situations.

\begin{figure}[tb] 
\centering 
\includegraphics[width=0.8\linewidth]{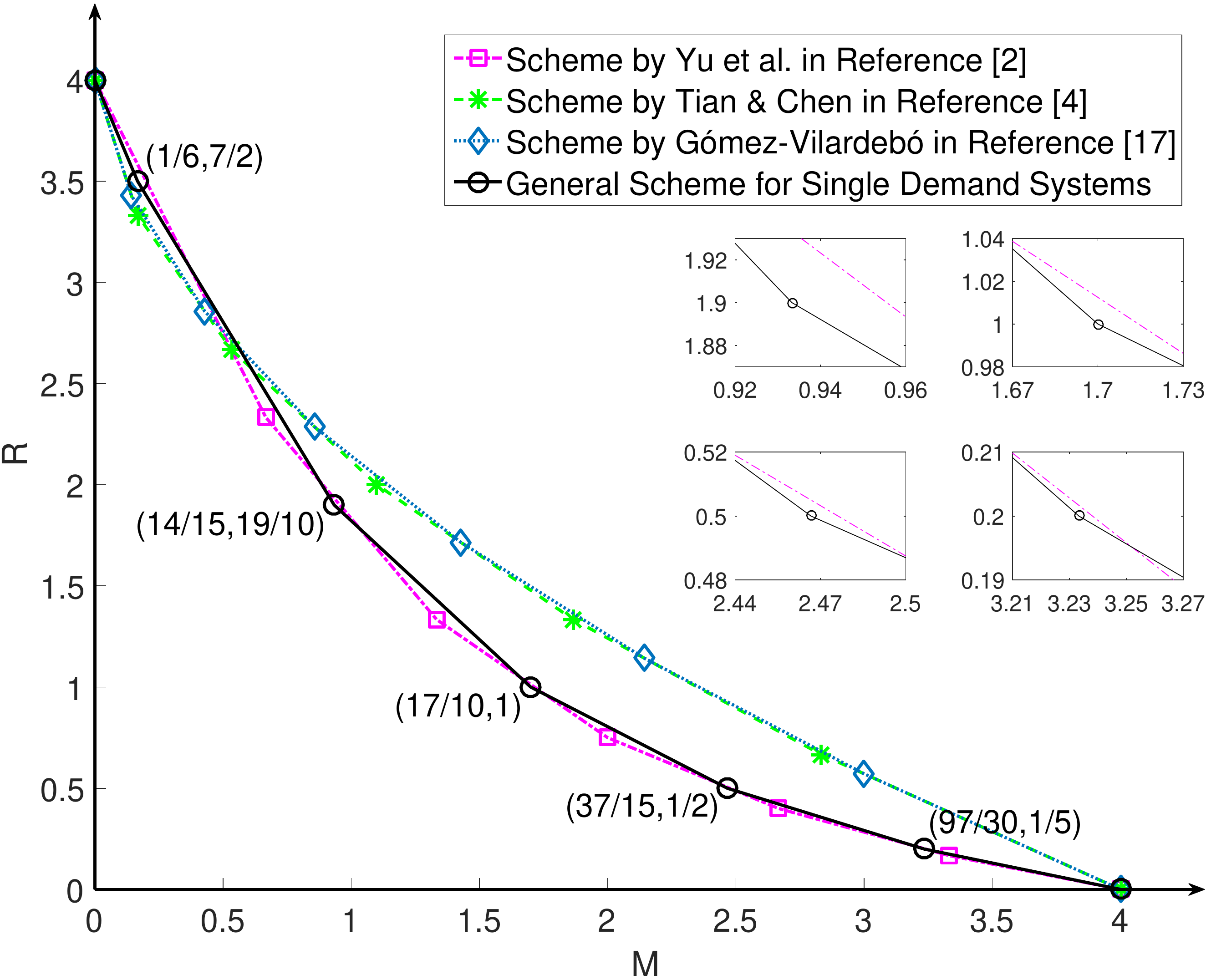}


\caption{Inner bounds for $(N,K)=(4,6)$ and demand type $(3,1,1,1)$.}
\label{fig:46}
\end{figure}

\section{Coding Scheme for Fully Demanded Coded Caching Systems}\label{general scheme}

The proposed code construction consists of multiple components which rely on each other to function as a whole. These components are: a novel file partition strategy, a prefetching strategy, a pairwise transformation, a delivery strategy, and a decoding strategy. The construction is parametrized by an integer $r$, whose meaning will become clear in the sequel. Throughout the section, we use $(N,K)=(3,6)$ and $r=1$ as our running example, and will provide details for this example whenever possible. 

\subsection{A Novel File Partition Strategy}

\begin{table*}[htbp]
\caption{File partition of $W_1$ in $(N,K)=(3,6)$ 
code caching system when $r=1$}
\label{table: file partition}
\centering
\scalebox{0.85}{
\scriptsize{
\begin{tabular}{|c|c|c|c|c|c|c|}
\hline
\diagbox[dir=SE]{$s$}{$\mathcal{R}$}& $\mathcal{R}=\{1\}$ &$\mathcal{R}=\{2\}$ & $\mathcal{R}=\{3\}$ & $\mathcal{R}=\{4\}$ & $\mathcal{R}=\{5\}$ & $\mathcal{R}=\{6\}$\\\hline
\multirow{3}{*}{\begin{tabular}[c]{@{}c@{}} $s=1$ \end{tabular}}
&
&\textcolor{blue}{$W^{I}_{1,\{2\},1}$, $W^{Q}_{1,\{2\},1}$} 
&\textcolor{blue}{$W^{I}_{1,\{3\},1}$, $W^{Q}_{1,\{3\},1}$}
&\textcolor{blue}{$W^{I}_{1,\{4\},1}$, $W^{Q}_{1,\{4\},1}$}
&\textcolor{blue}{$W^{I}_{1,\{5\},1}$, $W^{Q}_{1,\{5\},1}$}
&\textcolor{blue}{$W^{I}_{1,\{6\},1}$, $W^{Q}_{1,\{6\},1}$}\\
&  
&\textcolor{blue}{$W^{I}_{2,\{2\},1}$, $W^{Q}_{2,\{2\},1}$} 
&\textcolor{blue}{$W^{I}_{2,\{3\},1}$, $W^{Q}_{2,\{3\},1}$}
&\textcolor{blue}{$W^{I}_{2,\{4\},1}$, $W^{Q}_{2,\{4\},1}$}
&\textcolor{blue}{$W^{I}_{2,\{5\},1}$, $W^{Q}_{2,\{5\},1}$}
&\textcolor{blue}{$W^{I}_{2,\{6\},1}$, $W^{Q}_{2,\{6\},1}$}\\
&  
&\textcolor{blue}{$W^{I}_{3,\{2\},1}$, $W^{Q}_{3,\{2\},1}$} 
&\textcolor{blue}{$W^{I}_{3,\{3\},1}$, $W^{Q}_{3,\{3\},1}$}
&\textcolor{blue}{$W^{I}_{3,\{4\},1}$, $W^{Q}_{3,\{4\},1}$}
&\textcolor{blue}{$W^{I}_{3,\{5\},1}$, $W^{Q}_{3,\{5\},1}$}
&\textcolor{blue}{$W^{I}_{3,\{6\},1}$, $W^{Q}_{3,\{6\},1}$}\\
\hline
\multirow{3}{*}{\begin{tabular}[c]{@{}c@{}} $s=2$ \end{tabular}}
&\textcolor{red}{$W^{I}_{1,\{1\},2}$, $W^{Q}_{1,\{1\},2}$}
&
&$W^{I}_{1,\{3\},2}$, $W^{Q}_{1,\{3\},2}$ &$W^{I}_{1,\{4\},2}$, $W^{Q}_{1,\{4\},2}$ &$W^{I}_{1,\{5\},2}$, $W^{Q}_{1,\{5\},2}$ &$W^{I}_{1,\{6\},2}$, $W^{Q}_{1,\{6\},2}$\\
&\textcolor{red}{$W^{I}_{2,\{1\},2}$, $W^{Q}_{2,\{1\},2}$}
&
&$W^{I}_{2,\{3\},2}$, $W^{Q}_{2,\{3\},2}$ &$W^{I}_{2,\{4\},2}$, $W^{Q}_{2,\{4\},2}$ &$W^{I}_{2,\{5\},2}$, $W^{Q}_{2,\{5\},2}$ &$W^{I}_{2,\{6\},2}$, $W^{Q}_{2,\{6\},2}$\\
&\textcolor{red}{$W^{I}_{3,\{1\},2}$, $W^{Q}_{3,\{1\},2}$}
&
&$W^{I}_{3,\{3\},2}$, $W^{Q}_{3,\{3\},2}$ &$W^{I}_{3,\{4\},2}$, $W^{Q}_{3,\{4\},2}$ &$W^{I}_{3,\{5\},2}$, $W^{Q}_{3,\{5\},2}$ &$W^{I}_{3,\{6\},2}$, $W^{Q}_{3,\{6\},2}$\\
\hline
\multirow{3}{*}{\begin{tabular}[c]{@{}c@{}} $s=3$ \end{tabular}}
&\textcolor{red}{$W^{I}_{1,\{1\},3}$, $W^{Q}_{1,\{1\},3}$}
&$W^{I}_{1,\{2\},3}$, $W^{Q}_{1,\{2\},3}$
&
&$W^{I}_{1,\{4\},3}$, $W^{Q}_{1,\{4\},3}$ &$W^{I}_{1,\{5\},3}$, $W^{Q}_{1,\{5\},3}$ &$W^{I}_{1,\{6\},3}$, $W^{Q}_{1,\{6\},3}$\\
&\textcolor{red}{$W^{I}_{2,\{1\},3}$, $W^{Q}_{2,\{1\},3}$}
&$W^{I}_{2,\{2\},3}$, $W^{Q}_{2,\{2\},3}$
&
&$W^{I}_{2,\{4\},3}$, $W^{Q}_{2,\{4\},3}$ &$W^{I}_{2,\{5\},3}$, $W^{Q}_{2,\{5\},3}$ &$W^{I}_{2,\{6\},3}$, $W^{Q}_{2,\{6\},3}$\\
&\textcolor{red}{$W^{I}_{3,\{1\},3}$, $W^{Q}_{3,\{1\},3}$}
&$W^{I}_{3,\{2\},3}$, $W^{Q}_{3,\{2\},3}$
&
&$W^{I}_{3,\{4\},3}$, $W^{Q}_{3,\{4\},3}$ &$W^{I}_{3,\{5\},3}$, $W^{Q}_{3,\{5\},3}$ &$W^{I}_{3,\{6\},3}$, $W^{Q}_{3,\{6\},3}$\\
\hline
\multirow{3}{*}{\begin{tabular}[c]{@{}c@{}} $s=4$ \end{tabular}}
&\textcolor{red}{$W^{I}_{1,\{1\},4}$, $W^{Q}_{1,\{1\},4}$}
&$W^{I}_{1,\{2\},4}$, $W^{Q}_{1,\{2\},4}$
&$W^{I}_{1,\{3\},4}$, $W^{Q}_{1,\{3\},4}$
&
&$W^{I}_{1,\{5\},4}$, $W^{Q}_{1,\{5\},4}$ &$W^{I}_{1,\{6\},4}$, $W^{Q}_{1,\{6\},4}$\\
&\textcolor{red}{$W^{I}_{2,\{1\},4}$, $W^{Q}_{2,\{1\},4}$}
&$W^{I}_{2,\{2\},4}$, $W^{Q}_{2,\{2\},4}$
&$W^{I}_{2,\{3\},4}$, $W^{Q}_{2,\{3\},4}$
&
&$W^{I}_{2,\{5\},4}$, $W^{Q}_{2,\{5\},4}$ &$W^{I}_{2,\{6\},4}$, $W^{Q}_{2,\{6\},4}$\\
&\textcolor{red}{$W^{I}_{3,\{1\},4}$, $W^{Q}_{3,\{1\},4}$}
&$W^{I}_{3,\{2\},4}$, $W^{Q}_{3,\{2\},4}$
&$W^{I}_{3,\{3\},4}$, $W^{Q}_{3,\{3\},4}$
&
&$W^{I}_{3,\{5\},4}$, $W^{Q}_{3,\{5\},4}$ &$W^{I}_{3,\{6\},4}$, $W^{Q}_{3,\{6\},4}$\\
\hline
\multirow{3}{*}{\begin{tabular}[c]{@{}c@{}} $s=5$ \end{tabular}}
&\textcolor{red}{$W^{I}_{1,\{1\},5}$, $W^{Q}_{1,\{1\},5}$}
&$W^{I}_{1,\{2\},5}$, $W^{Q}_{1,\{2\},5}$
&$W^{I}_{1,\{3\},5}$, $W^{Q}_{1,\{3\},5}$
&$W^{I}_{1,\{4\},5}$, $W^{Q}_{1,\{4\},5}$
&
&$W^{I}_{1,\{6\},5}$, $W^{Q}_{1,\{6\},5}$\\
&\textcolor{red}{$W^{I}_{2,\{1\},5}$, $W^{Q}_{2,\{1\},5}$}
&$W^{I}_{2,\{2\},5}$, $W^{Q}_{2,\{2\},5}$
&$W^{I}_{2,\{3\},5}$, $W^{Q}_{2,\{3\},5}$
&$W^{I}_{2,\{4\},5}$, $W^{Q}_{2,\{4\},5}$
&
&$W^{I}_{2,\{6\},5}$, $W^{Q}_{2,\{6\},5}$\\
&\textcolor{red}{$W^{I}_{3,\{1\},5}$, $W^{Q}_{3,\{1\},5}$}
&$W^{I}_{3,\{2\},5}$, $W^{Q}_{3,\{2\},5}$
&$W^{I}_{3,\{3\},5}$, $W^{Q}_{3,\{3\},5}$
&$W^{I}_{3,\{4\},5}$, $W^{Q}_{3,\{4\},5}$
&
&$W^{I}_{3,\{6\},5}$, $W^{Q}_{3,\{6\},5}$\\
\hline
\multirow{3}{*}{\begin{tabular}[c]{@{}c@{}} $s=6$ \end{tabular}}
&\textcolor{red}{$W^{I}_{1,\{1\},6}$, $W^{Q}_{1,\{1\},6}$}
&$W^{I}_{1,\{2\},6}$, $W^{Q}_{1,\{2\},6}$
&$W^{I}_{1,\{3\},6}$, $W^{Q}_{1,\{3\},6}$
&$W^{I}_{1,\{4\},6}$, $W^{Q}_{1,\{4\},6}$
&$W^{I}_{1,\{5\},6}$, $W^{Q}_{1,\{5\},6}$
&\\
&\textcolor{red}{$W^{I}_{2,\{1\},6}$, $W^{Q}_{2,\{1\},6}$}
&$W^{I}_{2,\{2\},6}$, $W^{Q}_{2,\{2\},6}$
&$W^{I}_{2,\{3\},6}$, $W^{Q}_{2,\{3\},6}$
&$W^{I}_{2,\{5\},6}$, $W^{Q}_{2,\{5\},6}$
&$W^{I}_{2,\{6\},6}$, $W^{Q}_{2,\{6\},6}$
&\\
&\textcolor{red}{$W^{I}_{3,\{1\},6}$, $W^{Q}_{3,\{1\},6}$}
&$W^{I}_{3,\{2\},6}$, $W^{Q}_{3,\{2\},6}$
&$W^{I}_{3,\{3\},6}$, $W^{Q}_{3,\{3\},6}$
&$W^{I}_{3,\{4\},6}$, $W^{Q}_{3,\{4\},6}$
&$W^{I}_{3,\{5\},6}$, $W^{Q}_{3,\{5\},6}$
&\\
\hline
\end{tabular}
}}
\end{table*}

Let $r$ be a positive integer in the set $\{0,1,\ldots,K-1\}$, and let $|\mathcal{C}|$ be the cardinality of the set $\mathcal{C}$. In the proposed partition strategy, each file is partitioned into a total of $2(K-r)\binom{K}{r}$ segments, resulting in file segments denoted as $W^{A}_{f,\mathcal{R},s}$, through three steps. In the first step, each file $W_f$ is partitioned into $\binom{K}{r}$ equal-sized sub-files, resulting in the corresponding subscript label $\mathcal{R}\subset \mathcal{K}$ and $|\mathcal{R}|=r$.
In the second step, each such subfile is further partitioned into $(K-r)$ small subfiles of equal size, resulting in the corresponding subscript label $s\in \mathcal{K}\setminus \mathcal{R}$. In the third step, each small subfile is finally partitioned into two equal-sized segments, which are called the ``I channel" segment and ``Q channel" segments, resulting in the superscript label $A\in \{I,Q\}$. 

Consider our running example $(N,K)=(3,6)$ and $r=1$: here each file is partitioned into $2*(6-1)\binom{6}{1}=60$ segments, which are listed in Table \ref{table: file partition}.


\subsection{Prefetching}\label{prefetching}

In the prefetching phase, the I-channel and the Q-channel are prefetched in the same manner for each user $k\in \mathcal{K}$, independent of each other; for conciseness, we will use $A$ to denote either  $I$ or $Q$. 

\subsubsection{Uncoded content} The uncoded segments $W^{A}_{f,\mathcal{R},s}$ are placed into the cache memory of user $k$, for every $f\in\mathcal{F}$, every $\mathcal{R}$ such that $k\in \mathcal{R}$, every $s \notin \mathcal{R}$, and every $A\in \{I,Q\}$. A total of $m_1=2N\binom{ K-1}{r-1}( K-r)$ such segments are prefetched at user $k$. For our running example $(N,K)=(3,6)$ and $r=1$, the segments labeled red in Table \ref{table: file partition} are directly placed in the cache of user-$1$ in the uncoded form. 

\subsubsection{Coded content}

The coded content prefetched by user $k$ can be conveniently viewed as the parities of a product code \cite{Lin:01}: the segments with the same index $s=k$ (and the same channel) are placed in a matrix, where each row contains the segments of the same file index (i.e., the same $f$ index), and each column contains the segments of the same  $\mathcal{R}$ index.  For our running example, this matrix is shown in Table \ref{table: product code} in blue for user-1, which correspond to the elements in Table \ref{table: file partition}  that are also labeled blue. The coded symbols are then prefetched using this matrix as follows. 

\begin{table*}[htbp]
\centering
\caption{Designing encoded prefetching via a product code}
\label{table: product code}
\scalebox{0.6}{
\begin{tabular}{|c|c|c|c|c|c|c|}
\hline
 & \multirow{2}{*}{\begin{tabular}[c]{@{}c@{}} $\mathcal{R}=\{2\}$ \end{tabular}}  
 & \multirow{2}{*}{\begin{tabular}[c]{@{}c@{}} $\mathcal{R}=\{3\}$ \end{tabular}}
 & \multirow{2}{*}{\begin{tabular}[c]{@{}c@{}} $\mathcal{R}=\{4\}$ \end{tabular}}
 & \multirow{2}{*}{\begin{tabular}[c]{@{}c@{}} $\mathcal{R}=\{5\}$ \end{tabular}}
 & \multirow{2}{*}{\begin{tabular}[c]{@{}c@{}} $\mathcal{R}=\{6\}$ \end{tabular}}
 & Row \\
 & & & & & & parity check
 \\
\hline
\multirow{2}{*}{\begin{tabular}[c]{@{}c@{}} $f= 1$ \end{tabular}}
&\multirow{2}{*}{\begin{tabular}[c]{@{}c@{}}  \textcolor{blue}{$W^A_{1,\{2\},1}$} \end{tabular}}
&\multirow{2}{*}{\begin{tabular}[c]{@{}c@{}}  \textcolor{blue}{$W^A_{1,\{3\},1}$} \end{tabular}}
&\multirow{2}{*}{\begin{tabular}[c]{@{}c@{}}  \textcolor{blue}{$W^A_{1,\{4\},1}$} \end{tabular}}
&\multirow{2}{*}{\begin{tabular}[c]{@{}c@{}}  \textcolor{blue}{$W^A_{1,\{5\},1}$} \end{tabular}}
&\multirow{2}{*}{\begin{tabular}[c]{@{}c@{}}  \textcolor{blue}{$W^A_{1,\{6\},1}$} \end{tabular}}
&\dout{\textcolor{brown}{$Z^A_{1,\emptyset,1}=$}} 
\\ 
& & & & & &
\dout{\textcolor{brown}{$\bigoplus_{r\in [2:6]} W^A_{1,\{r\},1} $}}
\\
\hline
\multirow{2}{*}{\begin{tabular}[c]{@{}c@{}}  $f=2$ \end{tabular}}
&\multirow{2}{*}{\begin{tabular}[c]{@{}c@{}}  \textcolor{blue}{$W^A_{2,\{2\},1}$} \end{tabular}}
&\multirow{2}{*}{\begin{tabular}[c]{@{}c@{}}  \textcolor{blue}{$W^A_{2,\{3\},1}$} \end{tabular}}
&\multirow{2}{*}{\begin{tabular}[c]{@{}c@{}}  \textcolor{blue}{$W^A_{2,\{4\},1}$} \end{tabular}}
&\multirow{2}{*}{\begin{tabular}[c]{@{}c@{}}  \textcolor{blue}{$W^A_{2,\{5\},1}$} \end{tabular}}
&\multirow{2}{*}{\begin{tabular}[c]{@{}c@{}}  \textcolor{blue}{$W^A_{2,\{6\},1}$} \end{tabular}}
&\textcolor{brown}{$Z^A_{2,\emptyset,1}=$} \\ 
& & & & & &
\textcolor{brown}{$\bigoplus_{r\in [2:6]} W^A_{2,\{r\},1} $}\\
\hline
\multirow{2}{*}{\begin{tabular}[c]{@{}c@{}} $f = 3$ \end{tabular}}
&\multirow{2}{*}{\begin{tabular}[c]{@{}c@{}}  \textcolor{blue}{$W^A_{3,\{2\},1}$} \end{tabular}}
&\multirow{2}{*}{\begin{tabular}[c]{@{}c@{}}  \textcolor{blue}{$W^A_{3,\{3\},1}$} \end{tabular}}
&\multirow{2}{*}{\begin{tabular}[c]{@{}c@{}}  \textcolor{blue}{$W^A_{3,\{4\},1}$} \end{tabular}}
&\multirow{2}{*}{\begin{tabular}[c]{@{}c@{}}  \textcolor{blue}{$W^A_{3,\{5\},1}$} \end{tabular}}
&\multirow{2}{*}{\begin{tabular}[c]{@{}c@{}}  \textcolor{blue}{$W^A_{3,\{6\},1}$} \end{tabular}}
&\textcolor{brown}{$Z^A_{3,\emptyset,1}=$} \\ 
& & & & & &
\textcolor{brown}{$\bigoplus_{r\in [2:6]} W^A_{3,\{r\},1} $}\\
\hline
Column
& \textcolor{purple}{$Z^A_{\{2\},1}=$}
&\textcolor{purple}{$Z^A_{\{3\},1}=$}
&\textcolor{purple}{$Z^A_{\{4\},1}=$}
&\textcolor{purple}{$Z^A_{\{5\},1}=$}
&\textcolor{purple}{$Z^A_{\{6\},1}=$}
&\\
parity check
& \textcolor{purple}{$\bigoplus_{f\in [1:3]} W^A_{f,\{2\},1} $}
&\textcolor{purple}{$\bigoplus_{f\in [1:3]} W^A_{f,\{3\},1} $}
&\textcolor{purple}{$\bigoplus_{f\in [1:3]} W^A_{f,\{4\},1} $}
&\textcolor{purple}{$\bigoplus_{f\in [1:3]} W^A_{f,\{5\},1} $}
&\textcolor{purple}{$\bigoplus_{f\in [1:3]} W^A_{f,\{6\},1} $}
&\\
\hline
\end{tabular}
}
\end{table*}



\begin{itemize}
\item Column parity checks: These coded symbols $Z^{A}_{\mathcal{R},k}$'s are the parities for the column indexed by $\mathcal{R}$, i.e.,
\begin{align}
Z^{A}_{\mathcal{R},k}=\bigoplus_{f\in\mathcal{F}}W^{A}_{f,\mathcal{R},k}
\end{align}
for all possible set $\mathcal{R}$ satisfying $k\notin\mathcal{R}$. Clearly, there are $m_2=2\binom{ K-1}{r}$ such symbols, where the factor $2$ is due to the dual $I,Q$ channels.  

\item Row parity checks: 
For the row corresponding to file index $f$, there will be multiple row parities, and let us focus on one fixed $f$, i.e., one row. For each $\mathcal{R}^{-}\subseteq \mathcal{K} \backslash \{k\}$ where $|\mathcal{R}^{-}|=r-1$, a row parity $Z^{A}_{f,\mathcal{R}^{-},k}$ is given by 
\begin{align}
  Z^{A}_{f,\mathcal{R}^{-},k}=\bigoplus_{u\in\mathcal{ K\backslash}(\mathcal{R}^{-}\cup \{k\})} W^{A}_{f,\{u\}\cup\mathcal{R}^{-},k}.
\end{align}
Due to the inherent linear dependency in product codes, only some of these row parities need to prefetched. More precisely, let 
\begin{eqnarray*}
l_k & \triangleq & \begin{cases}
2 & \text{if $k=1$}\\
1 & \text{otherwise}\\
 \end{cases},
\end{eqnarray*}
then the prefetched row parities are $Z^{A}_{f,\mathcal{R}^{-},k}$ where $f\in\mathcal{F}\backslash \{1\}$ and $\mathcal{R}^{-}\subseteq \mathcal{K} \backslash \{k,l_k\}$ such that 
$\left|\mathcal{R}^{-}\right|=r-1$.
It is clear that the total number of prefetched row parities is $m_3=2(N-1)\binom{ K-2}{r-1}$. 


\end{itemize}


 The column and row parities at user-1 in our running example are given in Table \ref{table: product code}. The column parities are labeled purple, and the row parity bits are labeled brown, where $Z^A_{1,\emptyset,1}$ is not prefetched (strikethrough in the table). Since in this case $\mathcal{R}^{-}$ can only be the empty set, there is only one parity per row.

Collecting all the prefetched content, it is seen that the normalized cache size is
\begin{align}
M  =  \frac{m_1+m_2+m_3}{2K\binom{ K-1}{r}}=r\frac{N-1}{K-1}+\frac{r+1}{K}. 
\end{align}

The linear dependence among the coded parities is made more precise in the following lemma. 
\begin{lemma}\label{lemma: linear dependence}
Consider a fixed user $k$ and subset $\mathcal{R}^-_*\subset \mathcal{K}\setminus\{k\}$, where $l_k\in\mathcal{R}^-_*$ and $|\mathcal{R}^-_*|=r-1$, then for any $f\in \mathcal{F}$, 
\begin{align}
  Z^{A}_{f,\mathcal{R}^{-}_*,k} 
  =\bigoplus_{h\in\mathcal{K}\setminus(\{k\}\cup \mathcal{R}^-_*)} Z^{A}_{f, (\mathcal{R}^-_*\setminus\{l_k\})\cup\{h\},k}.\label{eq: linear dependence 1}
\end{align}
%
Moreover, we have 
\begin{align}
  Z^{A}_{1,\mathcal{R}^{-},k}=   \bigoplus_{\mathcal{R}: \mathcal{R}^{-}\subseteq \mathcal{R}, k\notin \mathcal{R}}
Z^{A}_{\mathcal{R},k}\oplus \bigoplus_{f\in \mathcal{F}\backslash\{1\}} Z^{A}_{f,\mathcal{R}^{-},k}.
  \label{eq: linear dependence 2}
\end{align}
\end{lemma}

The proof of this lemma can be found in Appendix \ref{Appendix remark}. Essentially, this lemma states that all the row parities can be computed using only the prefetched parities. Consider again our running example. The parity $Z^{A}_{1,\emptyset,1}$ can be computed according to \eqref{eq: linear dependence 2} in Lemma \ref{lemma: linear dependence} as 
\begin{align*}
  Z^{A}_{1,\emptyset,1}
= & \bigoplus_{t\in \{2,3,4,5,6\}} Z^{A}_{\{t\},1} \oplus \bigoplus_{f\in \{2,3\}}  Z^{A}_{f,\emptyset,1}.
\end{align*}

\subsection{Pairwise Transformation for Delivery}\label{transformation}

In our previous work \cite{SGZT-19} where $N=K$, a delivery strategy was provided to cancel certain interference symbols in the coded prefetched symbols to improve the delivery efficiency. 
However, when $K \geq N$ and some file $W_n$ is requested by an even number of users, this interference cancellation strategy would not succeed. In order to create the same cancellation effect in such cases, we have to create such an opportunity by increasing the dimension of the message space. This process is reminiscent of the interference alignment technique \cite{Jafar-08} in channel coding, where vector space techniques are essential.
The I/Q channel structure and the corresponding pairwise transformation are therefore introduced for this purpose, which we provide next. 


The delivery process will go through all $s\in\mathcal{K}$, and in the following we shall discuss a fixed $s$. For a demand $\bm{d}$, we denote $K_{f}$ as the number of users requesting file $W_{f}$. Correspondingly, denote the set of indices of these users as $\mathcal{K}_f$. For any $s\in\mathcal{K}$, denote its complement set as 
\begin{align}
\mathcal{C}_s \triangleq \mathcal{K}\backslash \{s\}. \label{eqn:U}
\end{align}
Among the users $\mathcal{C}_s$ that request file $W_f$, let the lowest indexed user be the \textit{leader} in $\mathcal{K}_f$, denoted as $u_{f}$.
The transformed segment pair $(W_{ d(t),\mathcal{R},s}^{(T),I},W_{ d(t),\mathcal{R},s}^{(T),Q})$ for $t,s\in\mathcal{K}$ is given by
\begin{align} \label{eq:defofTForW}
    \begin{bmatrix}W_{ d(t),\mathcal{R},s}^{(T),I}\\W_{ d(t),\mathcal{R},s}^{(T),Q} \end{bmatrix} =\mathbf{T}_{t,s}  \begin{bmatrix}W_{ d(t),\mathcal{R},s}^{I}\\W_{ d(t),\mathcal{R},s}^{Q} \end{bmatrix},
\end{align}
where $\mathbf{T}_{t,s}$ is given next. 
If $ d(t)= d(s)$, then define
\begin{eqnarray}
\mathbf{T}_{t,s} \triangleq \begin{cases}
\begin{bmatrix}1 & 0\\0 & 1 \end{bmatrix} 
 \vspace{0.3cm}
 & \text{if \ensuremath{K_{d(t)}\text{ is odd}}}\\
 \vspace{0.3cm}
 \begin{bmatrix}1 & 1\\1 & 0 \end{bmatrix}  & \text{if \ensuremath{K_{d(t)}\text{ is even and $t\neq s$}}}.\\
  \begin{bmatrix}0 & 1\\1 & 1 \end{bmatrix}  & \text{if \ensuremath{K_{d(t)}\text{ is even and $t= s$}}}
 \end{cases}\label{eq:1}
\end{eqnarray}
On the other hand,  if $ d(t)\neq d(s)$, define
 \begin{equation}
 \mathbf{T}_{t,s}\triangleq\begin{cases}
  \vspace{0.3cm}
 \begin{bmatrix}1 & 0\\0 & 1 \end{bmatrix} & \text{if \ensuremath{\ensuremath{\mathcal{K}_{ d(t)}}} is odd}\\
  \vspace{0.3cm}
 \begin{bmatrix}1 & 1\\1 & 0 \end{bmatrix} & \text{if \ensuremath{\mathcal{K}_{ d(t)}} is even and \ensuremath{t\neq u_{d(t)}}}\\
 \begin{bmatrix}0 & 1\\1 & 1 \end{bmatrix} & \text{if \ensuremath{\mathcal{K}_{ d(t)}} is even and \ensuremath{t=u_{d(t)}}}
 \end{cases}.\label{eq:2}
 \end{equation}
In other words, when file $W_{f}$ is requested by an odd number of users, the transformed pair  $(W_{f,\mathcal{R},s}^{(T),I},W_{f,\mathcal{R},s}^{(T),Q})$ are identical to the orginal segment pair $(W_{f,\mathcal{R},s}^{I},W_{f,\mathcal{R},s}^{Q})$; when file $W_{f}$ is requested by an even number of users, an invertible transformation is applied.

In our running example, when the demand vector is $(1,1,1,1,2,3)$, consider the case $t=2$, $s=1$, and $\mathcal{R}=\{3\}$. Since $d(t)=d(s)$ and $K_1=4$, the \textcolor{black}{transformed segments} are given as 
\begin{align*}
  &W_{d(t),\mathcal{R},s}^{(T),I}|_{t=2,s=1,\mathcal{R}=\{3\}}= W_{1,\{3\},1}^I\oplus W_{1,\{3\},1}^Q,\\
  &W_{d(t),\mathcal{R},s}^{(T),Q}|_{t=2,s=1,\mathcal{R}=\{3\}}= W_{1,\{3\},1}^I,
\end{align*}
by \eqref{eq:1}. Next consider another case where $t=2$, $s=5$, and $\mathcal{R}=\{3\}$. Since $d(s)\neq d(t)$, $K_{d(t)}=K_1=4$, and the leader of $\mathcal{K}_1$ is $u_{d(t)}=1$, we have $t\neq u_{d(t)}$. Hence, 
\eqref{eq:2} gives
\begin{align*}
  &W_{d(t),\mathcal{R},s}^{(T),I}|_{t=2,s=5,\mathcal{R}=\{3\}}= W_{1,\{3\},5}^I\oplus W_{1,\{3\},5}^Q,\\
  &W_{d(t),\mathcal{R},s}^{(T),Q}|_{t=2,s=5,\mathcal{R}=\{3\}}= W_{1,\{3\},5}^I.
\end{align*}
Table \ref{table: intermediate} shows all the transformed segments for $s=1$. 
\begin{table}[tb]
\centering
\caption{Transformed segments for $s=1$ \label{table: intermediate}}
\scalebox{0.9}{
\begin{tabular}{|c|c|c|}
\hline
$t\in \mathcal{K}\backslash \{1\}$  &  $W_{d(t),\mathcal{R},1}^{(T),I}$    &   $W_{d(t),\mathcal{R},1}^{(T),Q}$\\ \hline
2 & $W_{1,\mathcal{R},1}^{I}\oplus W_{1,\mathcal{R},1}^{Q}$& $W_{1,\mathcal{R},1}^{I}$ \\ \hline
3 & $W_{1,\mathcal{R},1}^{I}\oplus W_{1,\mathcal{R},1}^{Q}$& $W_{1,\mathcal{R},1}^{I}$  \\ \hline  
4 & $W_{1,\mathcal{R},1}^{I}\oplus W_{1,\mathcal{R},1}^{Q}$& $W_{1,\mathcal{R},1}^{I}$  \\ \hline
5 & $W_{2,\mathcal{R},1}^{I}$ & $W_{2,\mathcal{R},1}^{Q}$ \\ \hline
6 & $W_{3,\mathcal{R},1}^{I}$ & $W_{3,\mathcal{R},1}^{Q}$  \\ \hline
\end{tabular}}
\end{table}

\subsection{Delivery Strategy}
\label{delivery}
The delivery strategy has certain similarity to that in \cite{MAN-14} and \cite{Yu-18}. For an arbitrary $s\in \mathcal{K}$ and set $\mathcal{R}^{+} \subseteq \mathcal{K}\backslash \{s\}$ with $r+1$ users, i.e., $\left|\mathcal{R}^{+}\right|=r+1$, the message $Y^{A}_{\mathcal{R}^{+},s}$ is defined as 
\begin{eqnarray}\label{eq:delivery}
Y^{A}_{\mathcal{R}^{+},s} & = & \bigoplus_{t\in\mathcal{R}^{+}}W^{(T),A}_{ d(t),\mathcal{R}^{+}\backslash t,s} \label{allY}
\end{eqnarray}
for both $A=I$ and $Q$. In our running example with demand $\bm{d} = (1,1,1,1,2,3)$, these coded symbols for $s=1$ are shown in Table \ref{table:IQ2}. 


Due to the inherent linear dependence,  not all of the coded symbols given above are transmitted. Recall the definition of $\mathcal{C}_s$ in (\ref{eqn:U}), and we further define the leader set to be 
\begin{align*}
    \mathcal{L}_s = \{u_{d(t)}: t \in \mathcal{C}_s\}.
\end{align*}
If some $\mathcal{R}^+$ satisfies that $\mathcal{R}^+ \subseteq \mathcal{C}_s$ and $\mathcal{R}^+ \cap \mathcal{L}_s = \emptyset$,
then the corresponding encoded segments $Y^{A}_{\mathcal{R}^{+},s}$ are not transmitted for both $A=I$ and $Q$. This is similar to the improved delivery strategy in \cite{Yu-18}, and the following lemma makes this linear dependence among $Y^{A}_{\mathcal{R}^{+},s}$'s more precise. 
\begin{lemma}\label{lemma:redundancy}
For given $s\in \mathcal{K}$ and subset $\mathcal{B}$ such that $\mathcal{L}_s\subset \mathcal{B} \subseteq \mathcal{C}_s$ and $|\mathcal{B}|=|\mathcal{L}_s|+r+1$, let $\mathcal{V}_F$ be a superset of all subsets $\mathcal{V}$ of $\mathcal{B}$ such that each file is requested by exactly one user in $\mathcal{V}$. Then
\begin{align*}
\bigoplus_{\mathcal{V} \in \mathcal{V}_F} Y_{\mathcal{B} \backslash \mathcal{V},s}^I = 0, \quad\bigoplus_{\mathcal{V} \in \mathcal{V}_F} Y_{\mathcal{B} \backslash \mathcal{V},s}^Q = 0.
\end{align*}
\end{lemma} 
Lemma \ref{lemma:redundancy} can be proved in an almost identical manner as that in \cite[Lemma 1]{Yu-18}. The only difference is that in our setting, the dependence is considered for each fixed index $s\in \mathcal{K}$, and we operate on the transformed segments instead of the original segments. Therefore, we omit the details of the proof, and readers can refer to \cite{Yu-18} directly.

By Lemma \ref{lemma:redundancy}, for any non-leader users set $\mathcal{R}^+$ such that $\mathcal{R}^+ \subseteq \mathcal{C}_s$ and $\mathcal{R}^+ \cap \mathcal{L}_s = \emptyset$, $Y_{\mathcal{R}^{+},s}^I$ and $Y_{\mathcal{R}^{+},s}^Q$ can be computed using the transmitted symbols
\begin{align*}
    Y_{\mathcal{R}^{+},s}^I = \bigoplus_{\mathcal{V} \in \mathcal{V}_F \backslash \{\mathcal{L}_s\}} Y_{\mathcal{B} \backslash \mathcal{V},s}^I \\
    Y_{\mathcal{R}^{+},s}^Q = \bigoplus_{\mathcal{V} \in \mathcal{V}_F \backslash \{\mathcal{L}_s\}} Y_{\mathcal{B}\backslash \mathcal{V},s}^Q.
\end{align*}
We shall directly assume that all the symbols in (\ref{allY}) are available at users from here on.

 In our running example, for $s=1$ we have $\mathcal{C}_1=\{2,3,4,5,6\}$ and $\mathcal{L}_1=\{2,5,6\}$. If $\mathcal{R}^+=\{3,4\}$, obviously $\mathcal{R}^+\cap \mathcal{L}_s= \emptyset$. Hence the messages $(Y^I_{\{3,4\},1},Y^Q_{\{3,4\},1})$ are not transmitted, which is indicated by the strikethrough in Table \ref{table:IQ2}. These two symbols can be computed as
\begin{align*}
    Y^{A}_{\{3,4\},1}=Y^{A}_{\{2,3\},1}\oplus Y^{A}_{\{2,4\},1},\quad A\in \{I,Q\}.
\end{align*}

\begin{table}[tb]
\centering
\caption{Delivery for Demand $\bm{d} = (1,1,1,1,2,3)$, $s=1$}
\label{table:IQ2}
\scalebox{0.75}{
\begin{tabular}{|c|c|c|}
\hline
& I-channel & Q-channel                                                                          \\ \hline
$Y^{A}_{\{2,3\},1}$ & $W_{1,\{3\},1}^I\oplus W_{1,\{3\},1}^Q \oplus W_{1,\{2\},1}^I\oplus W_{1,\{2\},1}^Q$ &  $W_{1,\{3\},1}^I \oplus W_{1,\{2\},1}^I$                        \\ \hline
$Y^{A}_{\{2,4\},1}$ & $W_{1,\{4\},1}^I\oplus W_{1,\{4\},1}^Q \oplus W_{1,\{2\},1}^I\oplus W_{1,\{2\},1}^Q$ &  $W_{1,\{4\},1}^I \oplus W_{1,\{2\},1}^I$     \\ \hline
$Y^{A}_{\{2,5\},1}$ & $W_{1,\{5\},1}^I\oplus W_{1,\{5\},1}^Q \oplus W_{2,\{2\},1}^I$             &  $W_{1,\{5\},1}^I \oplus W_{2,\{2\},1}^Q$  \\ \hline
$Y^{A}_{\{2,6\},1}$ & $W_{1,\{6\},1}^I\oplus W_{1,\{6\},1}^Q \oplus W_{3,\{2\},1}^I$             &  $W_{1,\{6\},1}^I \oplus W_{3,\{2\},1}^Q$                       \\ \hline
*$Y^{A}_{\{3,4\},1}$ &  \dout{$W_{1,\{4\},1}^I\oplus W_{1,\{4\},1}^Q \oplus W_{1,\{3\},1}^I\oplus W_{1,\{3\},1}^Q$}             &  \dout{$W_{1,\{4\},1}^I \oplus W_{1,\{3\},1}^I$}     \\ \hline
$Y^{A}_{\{3,5\},1}$ & $W_{1,\{5\},1}^I\oplus W_{1,\{5\},1}^Q \oplus W_{2,\{3\},1}^I$             &  $W_{1,\{5\},1}^I \oplus W_{2,\{3\},1}^Q$     \\ \hline
$Y^{A}_{\{3,6\},1}$ & $W_{1,\{6\},1}^I\oplus W_{1,\{6\},1}^Q \oplus W_{3,\{3\},1}^I$             &  $W_{1,\{6\},1}^I \oplus W_{3,\{3\},1}^Q$                         \\ \hline
$Y^{A}_{\{4,5\},1}$ & $W_{1,\{5\},1}^I\oplus W_{1,\{5\},1}^Q \oplus W_{2,\{4\},1}^I$             &  $W_{1,\{5\},1}^I \oplus W_{2,\{4\},1}^Q$     \\ \hline
$Y^{A}_{\{4,6\},1}$ & $W_{1,\{6\},1}^I\oplus W_{1,\{6\},1}^Q \oplus W_{3,\{4\},1}^I$             &  $W_{1,\{6\},1}^I \oplus W_{3,\{4\},1}^Q$     \\ \hline
$Y^{A}_{\{5,6\},1}$ & $W_{2,\{6\},1}^I \oplus W_{3,\{5\},1}^I$             &  $W_{2,\{6\},1}^Q \oplus W_{3,\{5\},1}^Q$   \\ \hline
\end{tabular}}
\end{table}

For each $s \in \mathcal{K}, |\mathcal{C}_s| = K-1 $ and the number of leaders among users $\mathcal{C}_s$ is exactly $|\mathcal{L}_s|=N-1$ if user $s$ is the only user that request file $W_{d(s)}$ and $|\mathcal{L}_s|=N$ other wise. The number of skipped $Y_{\mathcal{R}^{+},s}^A$ symbols is $\binom{K-1-|\mathcal{L}_s|}{r+1}$. Hence, the communication rate $R$ is 
\begin{eqnarray}
R = \frac{T}{2K\binom{ K-1}{r}},
\end{eqnarray}
where 
\begin{align*}
T&=   2K\binom{K-1}{r+1} - 2\sum_{s \in \mathcal{K}}\binom{K-1-|\mathcal{L}_s|}{r+1}\\
&= 2K\binom{K-1}{r+1}-2(K-p_{\vec{\bm{e}}(\bm{d})})\binom{K-1-N}{r+1}\\
&-2p_{\vec{\bm{e}}(\bm{d})}\binom{K-1-(N-1)}{r+1}\\
&= 2K\binom{K-1}{r+1} -2K\binom{ K-1}{r} S(\vec{\bm{e}}(\bm{d}),r)
\end{align*}
matching the rate expression in Theorem \ref{thm1 iq channel}.

\subsection{Decoding}\label{decoding theory}
User $k$ needs $W_{ d(k),\mathcal{R},s}^A$ for all $s\in  \mathcal{K} $, $\mathcal{R} \subseteq  \mathcal{K} \backslash \{s\}, \left|\mathcal{R}\right|=r$, and $A \in \{I, Q\}$. Among all these segments, $W_{ d(k),\mathcal{R},s}^A$, where $k\in \mathcal{R}$, can be directly obtained from user-$k$'s cache in an uncoded form. The remaining required segments can be categorized into two classes:
\begin{enumerate} 
\item $(W_{d(k),\mathcal{R},s}^{I},W_{d(k),\mathcal{R},s}^{Q})$ pairs where $k\notin \mathcal{R}$ and $k\neq s$,
\item $(W_{d(k),\mathcal{R},k}^{I},W_{d(k),\mathcal{R},k}^{Q})$ pairs where $k\notin \mathcal{R}$ (i.e., $k=s$ in  $W_{d(k),\mathcal{R},s}$),
\end{enumerate}
which are mutually exclusive and will be decoded as follows.
 
\subsubsection{Decoding by individual symbol elimination}
We decode the first class of segments listed above using the same method as that in \cite{MAN-14}, but in the transformed domain. 
Notice that 
\begin{align}
&W_{ d(k),\mathcal{R},s}^{(T),A}\nonumber\\
  = & W_{ d(k),\mathcal{R},s}^{(T),A}
  \oplus\bigoplus_{j\in\mathcal{R}}W^{(T),A}_{ d(j), \mathcal{R}\cup \left\{k\right\} \backslash \{j\},s}
  \oplus\bigoplus_{i\in\mathcal{R}}W^{(T),A}_{ d(i), \mathcal{R}\cup\left\{ k\right\} \backslash \{i\},s}\nonumber\\
  = & \bigoplus_{j\in\mathcal{R}\cup k}W^{(T),A}_{ d(j), \mathcal{R}\cup\left\{ k\right\} \backslash \{j\},s}\oplus\bigoplus_{i\in\mathcal{R}}W^{(T),A}_{ d(i), \mathcal{R}\cup\left\{ k\right\} \backslash \{i\},s}\nonumber\\
  = & Y^{A}_{\mathcal{R}\cup\left\{ k\right\},s}\oplus\bigoplus_{i\in\mathcal{R}}W^{(T),A}_{ d(i), \mathcal{R}\cup\left\{ k\right\} \backslash \{i\},s},\quad A\in \{I,Q\}.\label{decode by cancellation}
\end{align}
Thus these transformed segments $W_{d(k),\mathcal{R},s}^{(T),A}$ can be decoded by using one delivered symbol $Y^{A}_{\mathcal{R}\cup\left\{ k\right\},s}$ and $r$ prefetched symbol $W^{(T),A}_{ d(i), \mathcal{R}\cup\left\{ k\right\} \backslash \{i\},s}$ (after applying the pairwise transformation). Since the pairwise transformation is invertible, user $k$ can indeed recover the first class of segments $(W_{ d(k),\mathcal{R},s}^{I},W_{ d(k),\mathcal{R},s}^{Q})$. 

In our running example, consider recovering the segments $(W_{ 1,\{2\},3}^I,W_{ 1,\{2\},3}^Q)$ at user 1, which are not prefetched. By substituting $\mathcal{R}=\{2\}$, $s=3$ and $k=1$ into \eqref{decode by cancellation}, we have
\begin{align*}
W_{ 1,\{2\},3}^{(T),I}=Y^I_{\{1,2\},3}\oplus W_{ 2,\{1\},3}^{(T),I},
\end{align*}
and
\begin{align*}
W_{ 1,\{2\},3}^{(T),Q}=Y^Q_{\{1,2\},3}\oplus W_{ 2,\{1\},3}^{(T),Q}.
\end{align*}
These two equations hold trivially since 
\begin{align*}
    Y^A_{\{1,2\},3}=W^{(T),A}_{1,\{2\},3}\oplus W^{(T),A}_{1,\{1\},3}
\end{align*}
for $A=\{I,Q\}$. Therefore, $(W_{ 1,\{2\},3}^{I},W_{ 1,\{2\},3}^{Q})$ can be decoded by inverting pairwise transformation.


\subsubsection{Decoding by interference alignment}

The second class of file segments to be recovered is those $(W_{d(k),\mathcal{R},k}^I,W_{d(k),\mathcal{R},k}^Q)$ at user $k$ where $\mathcal{R}\subseteq \mathcal{K}\backslash \{k\}$. The following lemma provides the key instrument for the decoding process, again in the transformed domain. 
\begin{lemma}\label{decoding intermediate}
For user $k$, $W_{ d(k),\mathcal{R},k}^{(T),A}$ can be computed as 
\begin{align}
W_{ d(k),\mathcal{R},k}^{(T),A}&=Z^{A}_{\mathcal{R},k}\oplus \bigoplus_{t\in\mathcal{R} }  Z^{(T),A}_{d(t),\mathcal{R}\backslash\{t\},k}\nonumber\\
&\oplus 
\bigoplus_{\mathcal{R}^{+}: \mathcal{R}\subset\mathcal{R}^{+},k\notin\mathcal{R}^{+}} Y^{A}_{\mathcal{R}^{+},k}  \label{decoding result}
\end{align}
for $A\in \{I,Q\}$ and $\mathcal{R}\subseteq  \mathcal{K}\backslash \{k\}$, where $Z^{(T),A}_{d(t),\mathcal{R}\backslash\{t\},k}$ is defined as 
\begin{align}
    \begin{bmatrix}Z^{(T),I}_{d(t),\mathcal{R}\backslash\{t\},k}\\Z^{(T),Q}_{d(t),\mathcal{R}\backslash\{t\},k}\end{bmatrix} \triangleq \mathbf{T}_{t,k}  \begin{bmatrix}Z^{I}_{d(t),\mathcal{R}\backslash\{t\},k}\\Z^{Q}_{d(t),\mathcal{R}\backslash\{t\},k}\end{bmatrix},
\end{align}
in which $\mathbf{T}_{t,k}$ is defined in \eqref{eq:1} and \eqref{eq:2}.
\end{lemma}

The proof of Lemma \ref{decoding intermediate} is given in Appendix \ref{decode}. By Lemma \ref{decoding intermediate}, any symbol in the second class can be written as an XOR of symbols on the right hand side of (\ref{decoding result}), each of which is either prefetched, can be computed, or delivered to user-$k$. Since the pairwise transformation is invertible, the second class of $(W_{d(k),\mathcal{R},k}^I,W_{d(k),\mathcal{R},k}^Q)$ can indeed be recovered from $(W_{ d(k),\mathcal{R},k}^{(T),I}, W_{ d(k),\mathcal{R},k}^{(T),Q})$. The right hand side of (\ref{decoding result}) can be viewed as aligning multiple interference signals $Z^{(T),A}_{d(t),\mathcal{R}\backslash\{t\},k}$ and $Y^{A}_{\mathcal{R}^{+},k}$ such that it can be removed from $Z^{A}_{\mathcal{R},k}$ to recover 
$W_{ d(k),\mathcal{R},k}^{(T),A}$. 

Consider in our running example the decoding process of the original segments $W_{1,\{2\},1}^{I}$ and $W_{1,\{2\},1}^{Q}$ at user 1. In this case $\mathcal{R}=\{2\}$ and $k=1$, and we decode the transformed segments $W_{1,\{2\},1}^{(T),A}$ for $A\in \{I,Q\}$. Since $t\in \mathcal{R}$, $t$ can only take value of $2$, hence $d(k)=d(t)$ and according to \eqref{eq:1} and \eqref{eq:2}, we have
\begin{align*}
& \begin{bmatrix} 
\bigoplus_{t\in\mathcal{R} }Z^{(T),I}_{d(t),\mathcal{R}\backslash\{t\},k} \\
\bigoplus_{t\in\mathcal{R} }Z^{(T),Q}_{d(t),\mathcal{R}\backslash\{t\},k}
\end{bmatrix} 
_{\mathcal{R}=\{2\},k=1}=
\begin{bmatrix} 
Z^{(T),I}_{d(2),\emptyset,1} \\
Z^{(T),Q}_{d(2),\emptyset,1}
\end{bmatrix} \\
= & \mathbf{T}_{2,1} 
\begin{bmatrix}
Z_{1,\emptyset,1}^I \\
Z_{1,\emptyset,1}^Q
\end{bmatrix}
= 
\begin{bmatrix} 
1 & 1 \\
1 & 0 \\
\end{bmatrix}
\begin{bmatrix}
Z_{1,\emptyset,1}^I \\
Z_{1,\emptyset,1}^Q
\end{bmatrix}
=
\begin{bmatrix}
Z_{1,\emptyset,1}^I \oplus Z_{1,\emptyset,1}^Q \\
Z_{1,\emptyset,1}^I
\end{bmatrix}.
\end{align*}
Moreover, since $\mathcal{R}^{+}\supset \mathcal{R}$, it can be $\{2,3\},\{2,4\},\{2,5\}$, or $\{2,6\}$. Therefore, by substituting all variables,  \eqref{decoding result} becomes 
\begin{align*}
W_{ 1,\{2\},1}^{(T),I}=  &Z^{I}_{\{2\},1}\oplus Z^{I}_{1,\emptyset,1} \oplus Z^{Q}_{1,\emptyset,1} \\
&\oplus Y^I_{\{2,3\},1} \oplus  Y^I_{\{2,4\},1} \oplus  Y^I_{\{2,5\},1}\oplus  Y^I_{\{2,6\},1} 
\end{align*}
and 
\begin{align*}
W_{ 1,\{2\},1}^{(T),Q}=  &Z^{Q}_{\{2\},1}\oplus Z^{I}_{1,\emptyset,1} \\
&\oplus Y^Q_{\{2,3\},1}\oplus  Y^Q_{\{2,4\},1} \oplus  Y^Q_{\{2,5\},1}\oplus  Y^Q_{\{2,6\},1},
\end{align*}
where $Z^{A}_{\{2\},1}$, $Z^{A}_{1,\emptyset,1}$ and $Y^A_{\mathcal{R}^+,1}$ can be found in Table \ref{table: product code} and Table \ref{table:IQ2}, from which the correctness of the two equations above can be verified. By inverting \eqref{eq:1} and \eqref{eq:2},  $(W_{ 1,\{2\},1}^{I},W_{ 1,\{2\},1}^{Q})$ can then be recovered.

\section{Conclusion}
In this paper, we consider coded caching systems with restricted demand types. We first showed that fully demanded systems may not be the worst case demand, contrary to popular beliefs. We then proposed a novel code construction for $(N,K)$ when $K \geq N$ and it is known a priori that all files are requested. The proposed code construction can achieve new operating corner points that are not covered by known constructions in the literature. 

\begin{appendices}



      \section{Proof of Lemma \ref{lemma: linear dependence}}\label{Appendix remark}


Define $\mathcal{R}_*^- \triangleq \mathcal{S} ^{-}\cup \{l_k\}$, where $|\mathcal{S} ^{-}|=r-2$. To prove (\ref{eq: linear dependence 1}), we can write 
\begin{align}
&Z^{A}_{f, \mathcal{R}_*^-,k}=\bigoplus_{h\in\mathcal{ K\backslash}\left(\mathcal{S}^{-}\cup\{l_k,k\}\right)}W^{A}_{f, \mathcal{S} ^{-}\cup \{l_k, h\},k}\nonumber\\
&\quad\overset{(a)}{=}\bigoplus_{h\in\mathcal{ K\backslash}\left(\mathcal{S}^{-}\cup\{l_k,k\}\right)}W^{A}_{f, \mathcal{S} ^{-}\cup \{l_k, h\},k}\nonumber\\
&\quad\qquad\oplus \bigoplus_{h_1,h_2\in \mathcal{K} \backslash\left(\mathcal{S}^{-}\cup\{l_k,k\}\right), h_1\neq h_2}W^{A}_{f, \mathcal{S} ^{-}\cup \{h_1,h_2\},k} \notag\\
&\quad=\bigoplus_{h\in\mathcal{ K\backslash}\left(\mathcal{S}^{-}\cup\{l_k,k\}\right)}W^{A}_{f, \mathcal{S} ^{-}\cup \{l_k,h\},k}\nonumber\\
&\quad\oplus \bigoplus_{h_1\in\mathcal{ K\backslash}\left(\mathcal{S}^{-}\cup\{l_k,k\}\right)}\bigoplus_{h\in\mathcal{ K\backslash}\left(\mathcal{S}^{-}\cup\{l_k,k,h_1\}\right)}W^{A}_{f, \mathcal{S} ^{-}\cup \{h_1,h_2\},k} \nonumber\\
&\quad=\bigoplus_{h_1\in\mathcal{ K\backslash}\left(\mathcal{S}^{-}\cup\{l_k,k\}\right)}\bigoplus_{h_2\in \mathcal{K} \backslash( \mathcal{S}^{-}\cup\{h_1,k\})} W^{A}_{f, \mathcal{S} ^{-}\cup \{h_1,h_2\},k} \nonumber\\
&\quad= \bigoplus_{h_1\in\mathcal{ K\backslash}\left(\mathcal{S}^{-}\cup\{l_k,k\}\right)} Z^{A}_{f, \mathcal{S} ^{-}\cup \{h_1\},k},
\end{align}
where $(a)$ is because 
\begin{align}
\bigoplus_{h_1,h_2\in \mathcal{K} \backslash\left(\mathcal{S}^{-}\cup\{l_k,k\}\right), h_1\neq h_2}W^{A}_{f, \mathcal{S} ^{-}\cup \{h_1,h_2\},k}=0
\end{align}
as each $W^{A}_{f, \mathcal{S} ^{-}\cup \{h_1\}\cup\{ h_2\},k}$ appears twice due to the symmetry between $h_1$ and $h_2$.

To prove (\ref{eq: linear dependence 2}), we can write 
\begin{align*}
 Z^{A}_{1,\mathcal{R}^{-},k}=&\bigoplus_{h\in\mathcal{ K\backslash}(\mathcal{R}^{-}\cup\{k\})} W^{A}_{1,\{h\}\cup\mathcal{R}^{-},k}\nonumber\\
=& \bigoplus_{f\in \mathcal{F}} \bigoplus_{h\in\mathcal{ K\backslash}(\mathcal{R}^{-}\cup\{k\})} W^{A}_{f,\{h\}\cup\mathcal{R}^{-},k}\nonumber\\
&\qquad\oplus \bigoplus_{f\in \mathcal{F}\backslash\{1\}} \bigoplus_{h\in\mathcal{ K\backslash}(\mathcal{R}^{-}\cup\{k\})} W^{A}_{f,\{h\}\cup\mathcal{R}^{-},k}\nonumber\\
=&  \bigoplus_{h\in\mathcal{ K\backslash}(\mathcal{R}^{-}\cup\{k\})}
\bigoplus_{f\in \mathcal{F}}W^{A}_{f,h\cup\mathcal{R}^{-},k}\nonumber\\
&\qquad\oplus \bigoplus_{f\in \mathcal{F}\backslash\{1\}} \bigoplus_{h\in\mathcal{ K\backslash}(\mathcal{R}^{-}\cup\{k\})} W^{A}_{f,h\cup\mathcal{R}^{-},k}\nonumber\\
=& \bigoplus_{h\in\mathcal{ K\backslash}(\mathcal{R}^{-}\cup\{k\})}
Z^{A}_{h\cup\mathcal{R}^{-},k}\oplus \bigoplus_{f\in \mathcal{F}\backslash\{1\}} Z^{A}_{f,\mathcal{R}^{-},k}\nonumber\\
=& \bigoplus_{\mathcal{R}:  \mathcal{R}^{-}\subset \mathcal{R}, k\notin \mathcal{R}}
Z^{A}_{\mathcal{R},k}\oplus \bigoplus_{f\in \mathcal{F}\backslash\{1\}} Z^{A}_{f,\mathcal{R}^{-},k}.
\end{align*}
The proof is complete. 
\IEEEQEDhere

\section{Proof of Lemma \ref{decoding intermediate}}\label{decode}
We need the following instrumental lemma in the proof. 
\begin{lemma}\label{z2}
For any $\bm{d}$, $s\in \mathcal{K}$, $\mathcal{R}\subseteq \mathcal{K}\backslash \{s\}$ where $|\mathcal{R}|=r$, and $A\in \{I,Q\}$, 
\begin{align}
 \bigoplus_{t\in \mathcal{K}} W^{(T),A}_{d(t),\mathcal{R},s}=   \bigoplus_{f\in \mathcal{F} } W^{A}_{f,\mathcal{R} ,s} =Z^{A}_{\mathcal{R},s}.  \label{eq: lemma for decoding} 
\end{align}
\end{lemma}
The proof of Lemma \ref{z2} can be found in Appendix \ref{appendix: transformed segments}.

\begin{IEEEproof}[Proof of Lemma \ref{decoding intermediate}]
We first write the following chain of equalities
\begin{align}
&W^{(T),A}_{ d(k),\mathcal{R},k}= \bigoplus_{s\in \mathcal{K} }W^{(T),A}_{d(s),\mathcal{R} ,k} \oplus  \bigoplus_{s\in \mathcal{K}\backslash\{k\} }W^{(T),A}_{d(s),\mathcal{R} ,k} \nonumber\\
&\quad\overset{(a)}{=} Z^{A}_{\mathcal{R},k} \oplus  \bigoplus_{v \in \mathcal{K}\backslash\{k\} }W^{(T),A}_{ d(v),\mathcal{R} ,k} \nonumber\\
&\quad\overset{(b)}{=} Z^{A}_{\mathcal{R},k} \oplus  \bigoplus_{v \in \mathcal{R} } W^{(T),A}_{ d(v),\mathcal{R} ,k}  \oplus  \bigoplus_{v \in \mathcal{K} \setminus (\mathcal{R}\cup \{k\})}W^{(T),A}_{ d(v),\mathcal{R} ,k} \nonumber\\
 &\qquad\quad\oplus  \bigoplus_{v \in \mathcal{K} \setminus (\mathcal{R}\cup \{k\})} \bigoplus_{t\in\mathcal{R} \cup \left\{ v \right\}} W^{(T),A}_{ d(t), \mathcal{R} \cup \left\{ v \right\} \backslash \{t\} ,k}\nonumber\\
 &\qquad\quad\oplus  \bigoplus_{v \in \mathcal{K} \setminus (\mathcal{R}\cup \{k\})} \bigoplus_{t\in\mathcal{R} \cup \left\{ v \right\}} W^{(T),A}_{ d(t), \mathcal{R} \cup \left\{ v \right\} \backslash \{t\} ,k}\nonumber\\
&\quad= Z^{A}_{\mathcal{R},k} \oplus  \bigoplus_{v \in \mathcal{R} } W^{(T),A}_{ d(v),\mathcal{R} ,k} \nonumber\\
 &\,\oplus  \bigoplus_{v \in \mathcal{K} \setminus (\mathcal{R}\cup \{k\})} \left(\bigoplus_{t\in\mathcal{R} \cup \left\{ v \right\}} W^{(T),A}_{ d(t), \mathcal{R} \cup \left\{ v \right\} \backslash \{t\} ,k}\oplus W^{(T),A}_{ d(v),\mathcal{R} ,k}\right)\nonumber\\
 &\qquad\quad\oplus  \bigoplus_{v \in \mathcal{K} \setminus (\mathcal{R}\cup \{k\})} \bigoplus_{t\in\mathcal{R} \cup \left\{ v \right\}} W^{(T),A}_{ d(t), \mathcal{R} \cup \left\{ v \right\} \backslash \{t\} ,k}\nonumber\\
&\quad = Z^{A}_{\mathcal{R},k} \oplus  \bigoplus_{v \in \mathcal{R} } W^{(T),A}_{ d(v),\mathcal{R} ,k} \nonumber\\
 &\qquad\quad\oplus  \bigoplus_{v \in \mathcal{K} \setminus (\mathcal{R}\cup \{k\})} \bigoplus_{t\in\mathcal{R} } W^{(T),A}_{ d(t), \mathcal{R} \cup \left\{ v \right\} \backslash \{t\} ,k}\nonumber\\
 &\qquad\quad\oplus  \bigoplus_{v \in \mathcal{K} \setminus (\mathcal{R}\cup \{k\})} \bigoplus_{t\in\mathcal{R} \cup \left\{ v \right\}} W^{(T),A}_{ d(t), \mathcal{R} \cup \left\{ v \right\} \backslash \{t\} ,k}\nonumber\\
&\quad\overset{(c)}{=} Z^{A}_{\mathcal{R},k} \oplus  \bigoplus_{t \in \mathcal{R} } W^{(T),A}_{ d(t),\mathcal{R} ,k} \nonumber\\
 &\qquad\quad\oplus   \bigoplus_{t\in\mathcal{R} } \bigoplus_{v \in \mathcal{K} \setminus (\mathcal{R}\cup \{k\})} W^{(T),A}_{ d(t), \mathcal{R} \cup \left\{ v \right\} \backslash \{t\} ,k}\nonumber\\
 &\qquad\quad\oplus  \bigoplus_{v \in \mathcal{K} \setminus (\mathcal{R}\cup \{k\})} \bigoplus_{t\in\mathcal{R} \cup \left\{ v \right\}} W^{(T),A}_{ d(t), \mathcal{R} \cup \left\{ v \right\} \backslash \{t\} ,k}\notag\\
 &\quad\overset{(d)}{=}Z^{A}_{\mathcal{R},k} \oplus  \bigoplus_{t\in\mathcal{R} }\bigoplus_{v\in\mathcal{K} \backslash (\mathcal{R}\backslash\{t\}\cup\{k\})} W^{(T),A}_{d(t),v\cup\mathcal{R}\backslash\{t\},k}\nonumber\\
 &\qquad\quad\oplus  \bigoplus_{v \in \mathcal{K} \setminus (\mathcal{R}\cup \{k\})} \bigoplus_{t\in\mathcal{R} \cup \left\{ v \right\}} W^{(T),A}_{ d(t), \mathcal{R} \cup \left\{ v \right\} \backslash \{t\} ,k}\notag\\
&\quad\overset{(e)}{=} Z^{A}_{\mathcal{R},k} \oplus  \bigoplus_{t\in\mathcal{R} }\bigoplus_{v\in\mathcal{K} \backslash (\mathcal{R}\backslash\{t\}\cup\{k\})} W^{(T),A}_{d(t),v\cup\mathcal{R}\backslash\{t\},k}\notag\\
&\qquad\quad\oplus \bigoplus_{\mathcal{R}^{+}: \mathcal{R}\subset \mathcal{R}^{+},k\notin \mathcal{R}^{+}} Y^{(T),A}_{\mathcal{R}^{+},k},\label{eqn:partlemma3}
\end{align}
where the steps can be justified as follows
\begin{enumerate}[label=(\alph*)]
\item Holds due to Lemma \ref{z2}; 
\item By splitting $\mathcal{K}\backslash\{ k\}$ into two mutually exclusive sets $\mathcal{R}$ and $\mathcal{K}\backslash (\mathcal{R}\cup \{k\})$, and then adding the same terms twice;
\item By exchanging the order of summations;
\item Due to the fact that for $t\in \mathcal{R}$
\begin{align*}
&\bigoplus_{v\in\mathcal{K} \backslash(\mathcal{R}\cup \left\{ k \right\})} W^{(T),A}_{ d(t), \mathcal{R} \cup \left\{ v \right\} \backslash \{t\} ,k}\oplus W^{(T),A}_{ d(t),\mathcal{R} ,k}\nonumber\\
&\quad=   \bigoplus_{v\in\mathcal{K} \backslash(\mathcal{R}\cup \left\{ k \right\})\cup \{t\}} W^{(T),A}_{ d(t), \mathcal{R} \cup \left\{ v \right\} \backslash \{t\} ,k}\nonumber\\
&\quad=  \bigoplus_{v\in\mathcal{K} \backslash (\mathcal{R}\backslash\{t\}\cup\{k\})}W^{(T),A}_{d(t),v\cup\mathcal{R}\backslash\{t\},k};
\end{align*}
\item By letting $\mathcal{R}^{+}=\mathcal{R} \cup \left\{ v \right\}$, from which it follows that
\begin{align*}
 & \bigoplus_{v \in \mathcal{K} \setminus (\mathcal{R}\cup \{k\})} \bigoplus_{t\in\mathcal{R} \cup \left\{ v \right\}} W^{(T),A}_{ d(t), \mathcal{R} \cup \left\{ v \right\} \backslash \{t\} ,k}\\
&\qquad=\bigoplus_{\mathcal{R}^{+}: \mathcal{R}\subset \mathcal{R}^{+},k\notin \mathcal{R}^{+}} \bigoplus_{t\in\mathcal{R}^{+}}W^{(T),A}_{ d(t),\mathcal{R}^{+}\backslash \{t\},k}\\
&\qquad=\bigoplus_{\mathcal{R}^{+}: \mathcal{R}\subset \mathcal{R}^{+},k\notin \mathcal{R}^{+}}Y^{(T),A}_{\mathcal{R}^{+},k}.
\end{align*}
\end{enumerate}

Comparing (\ref{decoding result}) and (\ref{eqn:partlemma3}), it is clear that it only remains to show that 
\begin{align}\label{eq:ZTA}
Z^{(T),A}_{d(t),\mathcal{R}\backslash\{t\},k}=\bigoplus_{v\in\mathcal{K} \backslash (\mathcal{R}\backslash\{t\}\cup\{k\})}W^{(T),A}_{d(t),v\cup\mathcal{R}\backslash\{t\},k}.
\end{align}  
To see this, note that by the definition of $ Z^{A}_{f,\mathcal{R}^{-},k}$
\begin{align}
    &\begin{bmatrix}Z^{(T),I}_{d(t),\mathcal{R}\backslash\{t\},k}\\Z^{(T),Q}_{d(t),\mathcal{R}\backslash\{t\},k}\end{bmatrix} =\mathbf{T}_{t,k}  \begin{bmatrix}Z^{I}_{d(t),\mathcal{R}\backslash\{t\},k}\\Z^{Q}_{d(t),\mathcal{R}\backslash\{t\},k}\end{bmatrix}\notag\\
    &\quad= \mathbf{T}_{t,k}\begin{bmatrix}
    \bigoplus_{u\in\mathcal{ K\backslash}(\mathcal{R}\setminus\{t\}\cup \{k\})} W^{I}_{d(t),\{u\}\cup(\mathcal{R}\setminus\{t\}),k}\\
    \bigoplus_{u\in\mathcal{ K\backslash}(\mathcal{R}\setminus\{t\}\cup \{k\})} W^{Q}_{d(t),\{u\}\cup(\mathcal{R}\setminus\{t\}),k}
    \end{bmatrix}\notag\\
    &\quad = \bigoplus_{u\in\mathcal{ K\backslash}(\mathcal{R}\setminus\{t\}\cup \{k\})} \left(\mathbf{T}_{t,k}\begin{bmatrix}W^{I}_{d(t),\{u\}\cup(\mathcal{R}\setminus\{t\}),k}\\
     W^{Q}_{d(t),\{u\}\cup(\mathcal{R}\setminus\{t\}),k}
    \end{bmatrix}\right)\notag\\
    &\quad = \bigoplus_{u\in\mathcal{ K\backslash}(\mathcal{R}\setminus\{t\}\cup \{k\})}\begin{bmatrix}W^{(T),I}_{d(t),\{u\}\cup(\mathcal{R}\setminus\{t\}),k}\\
    W^{(T),Q}_{d(t),\{u\}\cup(\mathcal{R}\setminus\{t\}),k}
    \end{bmatrix},
\end{align}
which completes the proof.\end{IEEEproof}

%

\section{Proof of Lemma \ref{z2}} \label{appendix: transformed segments}
We claim that for any $f\in \mathcal{F}$, when $K_f\neq \emptyset$, 
\begin{align}
\bigoplus_{t\in \mathcal{K}_f} W^{(T),A}_{d(t),\mathcal{R},s}  =   W^{A}_{f,\mathcal{R},s}, \label{fact}
\end{align}
from which the lemma follows trivially.

When $K_f$ is odd, \eqref{fact} holds true trivially. We only need to consider the case when $K_f$ is even. 
\begin{itemize}
\item $f=d(s)$ and $A=I$: 
\begin{align}
&\bigoplus_{t\in \mathcal{K}_{f}} W^{(T),I}_{d(t),\mathcal{R},s}=\bigoplus_{t\in \mathcal{K}_{f}} W^{(T),I}_{f,\mathcal{R},s} \notag\\
&\,\,=\bigoplus_{t\in \mathcal{K}_{f}\backslash \{s\}}\left( W^{I}_{f,\mathcal{R},s}\oplus W^{Q}_{f,\mathcal{R},s}\right)\oplus  W^{Q}_{f,\mathcal{R},s}\notag\\
&\,\,\overset{(a)}{=} W^{I}_{f,\mathcal{R},s}\oplus W^{Q}_{f,\mathcal{R},s}\oplus W^{Q}_{f,\mathcal{R},s}=W^{I}_{f,\mathcal{R},s},
\end{align}
where (a) is true because $K_f$ is even. 
\item $f=d(s)$ and $A=Q$:
\begin{align}
&\bigoplus_{t\in \mathcal{K}_{f}} W^{(T),Q}_{d(t),\mathcal{R},s} =\bigoplus_{t\in \mathcal{K}_{f}} W^{(T),Q}_{f,\mathcal{R},s} \notag\\
&\quad= \bigoplus_{t\in \mathcal{K}_{f}\backslash \{s\}} W^{I}_{f,\mathcal{R},s}\oplus  W^{I}_{f,\mathcal{R},s}\oplus  W^{Q}_{f,\mathcal{R},s}\notag\\
&\quad= W^{I}_{f,\mathcal{R},s}\oplus  W^{I}_{f,\mathcal{R},s}\oplus  W^{Q}_{f,\mathcal{R},s}= W^{Q}_{f,\mathcal{R},s},
\end{align}
for the same reason as in the previous case. 
\item $f\neq d(s)$ and $A=I$:
\begin{align}
&\bigoplus_{t\in \mathcal{K}_{f}} W^{(T),I}_{d(t),\mathcal{R},s}=\bigoplus_{t\in \mathcal{K}_{f}} W^{(T),I}_{f,\mathcal{R},s} \notag\\
&\,\,=\bigoplus_{t\in \mathcal{K}_{f}\backslash \{u_f\}}\left( W^{I}_{f,\mathcal{R},s}\oplus W^{Q}_{f,\mathcal{R},s}\right)\oplus  W^{Q}_{f,\mathcal{R},s}\notag\\
&\,\,= W^{I}_{f,\mathcal{R},s}\oplus W^{Q}_{f,\mathcal{R},s}\oplus W^{Q}_{f,\mathcal{R},s}=W^{I}_{f,\mathcal{R},s},
\end{align}
\item $f\neq d(s)$ and $A=Q$:
\begin{align}
&\bigoplus_{t\in \mathcal{K}_{f}} W^{(T),Q}_{d(t),\mathcal{R},s} =\bigoplus_{t\in \mathcal{K}_{f}} W^{(T),Q}_{f,\mathcal{R},s} \notag\\
&\quad= \bigoplus_{t\in \mathcal{K}_{f}\backslash \{u_f\}} W^{I}_{f,\mathcal{R},s}\oplus  W^{I}_{f,\mathcal{R},s}\oplus  W^{Q}_{f,\mathcal{R},s}\notag\\
&\quad= W^{I}_{f,\mathcal{R},s}\oplus  W^{I}_{f,\mathcal{R},s}\oplus  W^{Q}_{f,\mathcal{R},s}= W^{Q}_{f,\mathcal{R},s}.
\end{align}
\end{itemize}
Lemma \ref{z2} is thus proved. \IEEEQEDhere
\end{appendices}

\bibliographystyle{ieeetr}

\end{document}